\documentclass{article}

\usepackage{arxiv}
\usepackage[utf8]{inputenc}
\usepackage[T1]{fontenc}
\usepackage{graphicx}
\usepackage{multirow}
\usepackage{amsmath,amssymb,amsfonts}
\usepackage{amsthm}
\usepackage{mathrsfs}
\usepackage{xcolor}
\usepackage{textcomp}
\usepackage{manyfoot}
\usepackage{booktabs}
\usepackage{tabularx}
\usepackage{placeins}
\usepackage{pdflscape}
\usepackage{rotating}
\usepackage{array}
\usepackage{xurl}
\usepackage{algorithm}
\usepackage{algorithmic}
\usepackage{float}
\usepackage{microtype}
\usepackage[authoryear]{natbib}
\usepackage{authblk}
\usepackage[
  colorlinks=true,
  linkcolor=blue,
  citecolor=blue,
  urlcolor=blue
]{hyperref}
\usepackage{doi}
\setcitestyle{aysep={,}}

\renewcommand{\headeright}{arXiv Preprint}
\renewcommand{\undertitle}{arXiv Preprint}
\renewcommand{\shorttitle}{H-RegAd}

\theoremstyle{plain}
\newtheorem{theorem}{Theorem}
\newtheorem{lem}{Lemma}

\theoremstyle{remark}
\newtheorem{rmk}{Remark}

\theoremstyle{definition}
\newtheorem{assumption}{Assumption}
\newtheorem{definition}{Definition}

\newenvironment{pf}{\begin{proof}}{\end{proof}}

\newcommand{\abs}[1]{\left\lvert #1 \right\rvert}
\newcommand{\inner}[2]{\left\langle #1,\ #2 \right\rangle}

\title{Robust Log-Contrast Regression for High-Dimensional Compositional Microbiome Data with Outliers}

\author[1]{Xuke Hua\thanks{Email: \texttt{xkhua@henu.edu.cn}.}}
\author[2]{Yunhai Xiao\thanks{Email: \texttt{yhxiao@henu.edu.cn}.}}
\author[1]{Xin Xin\thanks{Corresponding author. Email: \texttt{xinxin@henu.edu.cn}. ORCID: \href{https://orcid.org/0009-0002-5933-4083}{0009-0002-5933-4083}.}}
\affil[1]{School of Mathematics and Statistics, Henan University, Kaifeng 475000, China}
\affil[2]{Center for Applied Mathematics of Henan Province, Henan University, Zhengzhou 450046, China}
\date{}

\hypersetup{
  pdftitle={Robust Log-Contrast Regression for High-Dimensional Compositional Microbiome Data with Outliers},
  pdfauthor={Xuke Hua, Yunhai Xiao, Xin Xin},
  pdfkeywords={Compositional microbiome data, Huber loss, Log-contrast regression, Non-asymptotic property, Alternating direction method of multipliers}
}

\begin{document}
\maketitle

\begin{abstract}
It is well known that microbiome sequencing data are naturally compositional and often have a complex subcomposition structure, as well as a lot of noise and some outliers. These characteristics make classical log-contrast regression difficult to apply, because it is sensitive to extreme values and cannot fully meet the linear constraints of compositional predictors.
To tackle these issues, in this paper, we propose a robust log-contrast regression framework that incorporates the Huber loss function along with a Lasso penalty, which not only boosts resistance to outliers but also facilitates sparse variable selection.
It is important to note that the general linear constraints, which are explicitly imposed, serve to preserve the algebraic coherence of compositional and subcompositional covariates.
Under some mild assumptions, we theoretically prove that the proposed method achieves non-asymptotic guarantees in finite samples, comparable to those of Huber loss estimators, ensuring robustness.
To implement the method efficiently, we develop a symmetric Gauss--Seidel-based alternating direction method of multipliers based on its dual formulation with a convergence guarantee, and it is also easy to implement in the sense that each subproblem admits a closed-form solution.
To evaluate the statistical performance of the proposed method, we conduct a series of numerical experiments using simulated data and gut microbiome data, which show that the proposed method achieves better estimation accuracy, competitive prediction performance, and more reliable outlier identification than the benchmark method RobRegCC.
\end{abstract}

	
	\section{Introduction}
	It is well known that compositional data are common in many areas of science, such as the relative abundances of microbial taxa in the human microbiome, the chemical composition of rocks and sediments in geology, and abundance profiles in microbial ecology.
	In this context, high-throughput sequencing techniques like target amplicon sequencing and metagenomic sequencing have made it easier to study microbial communities \citep{Callahan2016DADA2}.
	Specifically, these techniques facilitate the characterization of microbial communities in various environments from the marine ecosystem  to the human gut.
	Furthermore, with the help of modern bioinformatics tools, sequencing reads are usually processed into sparse compositional profiles that reflect the relative abundances of microbial groups \citep{Caporaso2010QIIME,callahan2017exact}. Therefore, these compositional profiles have become an important basis for subsequent statistical analysis in many ecological and biomedical studies.
	
	In many microbiome studies, these compositional data are analyzed together with environmental or host-related covariates, such as soil physicochemical properties, clinical or physiological measurements, metabolomics data, and flow cytometry profiles. 
	To characterize the composition, sequencing reads are often aligned to reference databases using bioinformatics tools such as MetaPhlAn \citep{Segata2012MetaPhlAn}. 
	Because the amount of extractable DNA can vary substantially across samples, the total number of reads often differs considerably.
	To address this issue, researchers  commonly normalize sequencing data and convert it into the relative abundances of microbial taxa within each sample. 
	We note that this preprocessing step turns the data into high-dimensional compositions with a unit-sum constraint, so that only relative information is retained. 
	However, in practice, many widely used metagenomic profiling tools, such as MEGAN \citep{Huson2007MEGAN} and MetaPhlAn \citep{Segata2012MetaPhlAn}, mainly provide relative abundance profiles at various taxonomic levels.
	As a result, statistical analysis of microbiome data must account for the special structure of compositional data, instead of treating them as ordinary multivariate data.
	
	In microbiome research, a key question is how microbial compositions are related to environmental factors or host traits. 
	However, because microbiome data contain only relative abundance rather than absolute counts, the standard linear regression approach cannot be used directly. To address this issue, a commonly used approach is the log-contrast regression model, which was introduced by \cite{Aitchison1984LogContrast} in the context of mixture experiments.
	In this approach, a continuous response is modeled as a linear combination of the log-transformed compositional components, with the regression coefficients constrained to sum to zero. As a result, this model provides a natural way to study the relationship between microbial composition and environmental or host-related variables.
	Furthermore, in microbiome studies, the log-contrast regression has been extended to high-dimensional settings, such as $\ell_{1}$-penalized methods for variable selection \citep{Lin2014CompositionalSelection} and other flexible structured approaches \citep{James2020PenalizedOptimization,Shi2016MicrobiomeRegression,Sun2020FunctionalLogContrast,Mishra2022RobustCompositional,Monti2024RobustKnockoff,Rios2025COCRegression, monti2025robust, zeng2026penalized}.
	
	From a biological point of view, it is often not enough to study only the overall composition, as the subcomposition structure within higher-level taxonomic groups can also be highly informative. For example, within a given phylum or genus, the relative proportions of the corresponding lower-level taxa may reflect important ecological features, such as community stability, functional redundancy, or early signs of dysbiosis.
	Therefore, understanding how these subcompositional patterns relate to phenotypic or environmental outcomes is important for revealing finer-scale microbial mechanisms \citep{Shi2016MicrobiomeRegression}.
	It should be emphasized that, after normalization, each subcomposition still sums to one and remains compositional in nature. As a result, the analysis of subcompositions must still respect the simplex structure and employ models that properly account for the constraints of compositional microbiome data.
	
	Practically, microbiome sequencing data are often affected by uneven sequencing depth, contamination, instrumental noise, and high proportions of zeros, which may give rise to extreme or outlier observations.
	In such cases, using the standard least-squares loss function may lead to materially distorted estimates, as the squared error is highly sensitive to outliers, see e.g.,  \cite{Shi2016MicrobiomeRegression} and \cite{Wang2017SubcompositionSelection}.
	Additionally, as noted by \cite{Mishra2022RobustCompositional} that, this sensitivity can impact not only coefficient estimation but also variable selection.
	In this paper, motivated by the noise and occasional outliers commonly present in compositional microbiome data, we incorporate the Huber loss function \citep{Huber1964RobustLocation} into subcomposition regression to enhance the robustness of the estimation.
	It is well known that the key feature of the Huber loss function is that it retains the efficiency of the squared loss for moderate errors while reducing the influence of extreme observations, which leads to more robust estimation and more reliable variable selection.
	Additionally, we incorporate a Lasso penalty \citep{tibshirani1996regression} for variable selection and provide theoretical guarantees showing that the proposed estimation method achieves strong non-asymptotic performance in finite samples under mild assumptions.
	To practically implement the robust log-contrast regression method, we employ the symmetric Gauss--Seidel--based alternating direction method of multipliers (sGS-ADMM) \citep{Li2018BlockSymmetricGS}, which is convergence-guaranteed and also easily implementable in the sense that all involved subproblems have closed-form solutions.
	It is worth noting that the proposed sGS-ADMM, based on the dual formulation, is generally efficient for tackling high-dimensional statistical problems. To assess the statistical efficiency of the proposed estimation method and the computational efficiency of the sGS-ADMM algorithm, we conduct a series of numerical experiments using both simulated data and real-world gut microbiome data. Additionally, we compare the performance of our method, named H-RegAd, with the representative robust method RobRegCC, as introduced in \cite{Mishra2022RobustCompositional}.
	The numerical results demonstrate that the proposed method achieves superior estimation accuracy, more competitive predictive performance, and more reliable outlier identification than RobRegCC.
	
	The remainder of this paper is organized as follows.  
	In Section \ref{sec:method}, we review existing log-contrast regression methods with linear constraints and introduce a robust log-contrast regression approach that incorporates the Huber loss function and a Lasso penalty.  
	In Section \ref{sec:framework}, we present the sGS-ADMM algorithm, which is practically implemented for Lasso penalized Huber loss function under general linear compositional constraint.  
	In Section \ref{sec:analysis}, we investigate the non-asymptotic statistical properties of the proposed method under certain regularity conditions.  
	In Section \ref{sec:simulation}, we report the results of numerical experiments conducted on simulated data.  
	In Section \ref{sec:realdata}, we demonstrate the practical application of the method using a real-world gut microbiome dataset.  
	Finally, in Section \ref{sec:discussion}, we conclude the paper with some discussions and remarks.

	\section{Robust Log-Contrast Regression}
	\label{sec:method}
	\subsection{Standard Log-Contrast Regression}
	To formulate the log-contrast regression, let $y \in \mathbb{R}^n$ denote the response vector, and let $X = (x_{ij}) \in \mathbb{R}^{n \times p}$ represent the matrix of compositional covariates, where $x_{ij}$ corresponds to the relative abundance of the $j$-th component in the $i$-th sample. 
	The $i$-th row of $X$, denoted by $x_i = (x_{i1}, \ldots, x_{ip})^\top$, is a $p$-dimensional composition where each $x_{ij} \geq 0$ for $j = 1, \ldots, p$ and $\sum_{j=1}^p x_{ij} = 1$, meaning that each $x_i$ represents relative information.
	Moreover, since zero components often appear in compositional data, they are typically substituted with a small pseudo-count, such as $0.5$ by \cite{Aitchison1982CoDA}.
	Under this setup, each compositional vector $x_i$ lies in the $(p-1)$-dimensional simplex defined as 
	$S^{p-1} \triangleq \{ (x_1, \ldots, x_p) \mid x_j \geq 0, \sum_{j=1}^p x_{j} = 1 \}$.
	Let $\beta \in \mathbb{R}^{p}$ denote the regression coefficient vector, and let the error 
	vector be $\epsilon = (\epsilon_1,\ldots,\epsilon_n)^\top$, where the $\epsilon_i$ are independent and identically distributed (i.i.d.)
	with mean $E(\epsilon_i) = 0$ and variance $\mathrm{Var}(\epsilon_i) = \sigma^2$. 
	The standard linear regression model is then written as
	\begin{equation}\label{eq:standard_model}
		y = X\beta + \epsilon .
	\end{equation}
	It should be noted that directly applying the standard linear regression model \eqref{eq:standard_model} to compositional microbiome data is generally inappropriate. Since the components of each compositional vector are constrained to unit-sum, the predictors are inherently dependent and may exhibit perfect multicollinearity, particularly when an intercept term is included.  
	
	To overcome the limitations of standard linear regression, \cite{Aitchison1984LogContrast} proposed the log-contrast regression framework, where the response variable is modeled using log-ratios relative to a chosen reference component, specifically:
	\begin{equation}\label{eq:logcontrast}
		y = Z_{\setminus k}\beta_{\setminus k} + \epsilon ,
	\end{equation}
	where $Z_{\setminus k}:=\left\{\log\left({x_{ij}}/{x_{ik}}\right):j\neq k\right\}\in \mathbb{R}^{n\times (p-1)}$
	is the log-ratio design matrix formed by taking the $k$-th component as the reference, and
	$\beta_{\setminus k}:=(\beta_1,\ldots,\beta_{k-1},\beta_{k+1},\ldots,\beta_p)^\top\in \mathbb{R}^{p-1}$.
	Clearly we see that although the formulation in \eqref{eq:logcontrast} offers a natural approach to modeling relative information, it relies on the selection of a reference component, and choosing different reference components can result in varying numerical representations.
	To avoid the dependence on an arbitrarily chosen reference component,
	\cite{Lin2014CompositionalSelection} reformulated the log-contrast model
	by introducing a full coefficient vector
	$\beta=(\beta_1,\ldots,\beta_p)^\top$ together with a zero-sum constraint.
	Indeed, taking the $k$-th component as the reference gives
	\[
		\sum_{j\neq k}\beta_j\log\left(\frac{x_{ij}}{x_{ik}}\right)
		=\sum_{j\neq k}\beta_j\log x_{ij}
		-\left(\sum_{j\neq k}\beta_j\right)\log x_{ik}.
	\]
	By setting	$\beta_k=-\sum_{j\neq k}\beta_j$,
	the linear predictor becomes	$\sum_{j=1}^{p}\beta_j\log x_{ij}$, with	$\sum_{j=1}^{p}\beta_j=0$. Thus, the zero-sum constraint naturally	arises from the reference-based log-ratio representation and removes	the dependence on the choice of reference component. Under this	representation, model~\eqref{eq:logcontrast} can be rewritten as
	\begin{equation}\label{eq:constrained_reg}
		y = Z\beta + \epsilon,	\qquad	\mathbf{1}_p^{\top}\beta = 0,
	\end{equation}
	where
	$Z=(z_{ij})=(\log x_{ij})\in\mathbb{R}^{n\times p}$, and $\mathbf{1}_p$ is a $p$-dimensional vector whose components are all equal
	to one. It is worth pointing out that this constrained approach treats all components equally and makes sure our fitted model depends only on the log-ratios of the original components, keeping the essential structure of compositional microbiome data.
	Moreover, in microbiome studies, the number of microbial features often greatly exceeds the sample size, making variable selection essential for identifying taxa linked to the response. In this context, \cite{Lin2014CompositionalSelection} introduced a Lasso-penalized \citep{tibshirani1996regression} estimation method for variable selection, which is formulated as follows:
	\begin{equation}\label{eq:l1_constrained}
		\min_{\beta\in\mathbb{R}^p}
		\left\{
		\frac{1}{2n}\|y-Z\beta\|_2^2
		+
		\lambda\|\beta\|_1
		\right\},
		\quad
		\text{s.t.}
		\quad
		\mathbf{1}_p^\top\beta=0,
	\end{equation}
	where $\lambda>0$ is a tuning parameter. 
	This estimator follows the compositional constraint and allows for sparse variable selection in high-dimensional regression. However, it is still based on squared-error loss and does not directly consider more complex structural information, such as the taxonomic or phylogenetic relationships among microbial taxa.
	\subsection{Subcomposition Log-Contrast Regression }
	
	To move beyond the zero-sum constraint in
	\eqref{eq:l1_constrained},
	\cite{Shi2016MicrobiomeRegression} incorporated the hierarchical
	structure of microbial taxa into the log-contrast regression framework.
	Specifically, in microbiome studies, microbial features are naturally
	organized into a phylogenetic or taxonomic tree, and this structure
	induces multiple linear constraints corresponding to different
	subcompositions. Building on this idea,
	\cite{Shi2016MicrobiomeRegression} proposed the following sparse
	regression model:
	\begin{equation}\label{eq:tree_constrained}
		\min_{\beta\in\mathbb{R}^p}
		\left\{\frac{1}{2n}\|y-Z\beta\|_{2}^{2}+\lambda\|\beta\|_{1}\right\},
		\qquad \text{s.t.}\qquad C^{\top}\beta=0.
	\end{equation}
	where the constraint matrix takes the following form
	\begin{equation}\label{defC}
			C^{\top}
			=[c_{1},c_{2},\ldots,c_{m}]^{\top}
			=
			\begin{pmatrix}
				\mathbf{1}_{p_1}^{\top} & 0 & \cdots & 0\\
				0 & \mathbf{1}_{p_2}^{\top} & \cdots & 0\\
				\vdots & \vdots & \ddots & \vdots\\
				0 & 0 & \cdots & \mathbf{1}_{p_m}^{\top}
			\end{pmatrix}_{m\times p}.
	\end{equation}
	Suppose that the $p$ microbial features are divided into $m$	subcompositions, where the $j$-th subcomposition contains $p_j$	features and $\sum_{j=1}^{m}p_j=p$. Accordingly, the coefficient vector	can be partitioned as	$\beta=	\left(	\beta_{(1)}^{\top},	\beta_{(2)}^{\top},	\ldots,	\beta_{(m)}^{\top}	\right)^{\top}	$ with	$\beta_{(j)}\in\mathbb{R}^{p_j}	$.
	To preserve subcompositional coherence, a separate zero-sum constraint	is imposed on the coefficients within each subcomposition $\mathbf{1}_{p_j}^{\top}\beta_{(j)}=0$ for $j=1,\ldots,m$. 	
	Collecting these $m$ constraints gives $C^{\top}\beta=0$. Specifically,
	\[
	C^{\top}\beta	=
	\begin{pmatrix}
		\mathbf{1}_{p_1}^{\top}\beta_{(1)}\\[2mm]
		\mathbf{1}_{p_2}^{\top}\beta_{(2)}\\[2mm]
		\vdots\\[2mm]
		\mathbf{1}_{p_m}^{\top}\beta_{(m)}
	\end{pmatrix}
	= \textbf{0}_{m},
	\]
	where $\textbf{0}_{m}$ is a $m$-dimensional zero vector. Therefore, the constraint $C^{\top}\beta=0$ requires the coefficients	within each subcomposition to sum to zero. Clearly, this formulation extends the standard log-contrast regression model~\eqref{eq:l1_constrained} by replacing the single global zero-sum
	constraint with multiple group-specific zero-sum constraints, thereby preserving the subcompositional structure within the taxonomic hierarchy.
	Moreover, to mitigate the bias introduced by the Lasso penalty, \cite{Shi2016MicrobiomeRegression} developed a debiasing procedure for statistical inference, the details of which are omitted here.
	
	It should be noted that there is other related work on incorporating hierarchical structures. For example, \cite{Wang2017SubcompositionSelection} proposed a tree-guided variable selection method that promotes structured sparsity along the taxonomic tree. More precisely, 
	\begin{equation}\label{eq:tree_guided}
		\min_{\beta\in\mathbb{R}^p}
		\left\{\frac{1}{2n}\|y-Z\beta\|_{2}^{2}
		+\lambda\|J_{\alpha}\beta\|_{1}\right\},
		\qquad \text{s.t.}\qquad \mathbf{1}_{p}^{\top}\beta=0.
	\end{equation}
	where $J_{\alpha}$ is a tree-structured coding matrix that promotes sparsity patterns consistent with the taxonomic hierarchy, and the tuning parameter $\alpha \in [0,1]$ controls the strength of the structural penalty.
	Additionally, \cite{Sun2020FunctionalLogContrast} extended the log-contrast framework to longitudinal microbiome studies by introducing functional compositional predictors, which integrate compositional regression with temporal effects and yield a group-Lasso-penalized regression model:
	\begin{equation}\label{eq:temporal_logcontrast}
		\min_{\beta,\beta_0,\beta_c}
		\left\{\frac{1}{2n}\big\|y-\beta_0\mathbf{1}_n-Z_c\beta_c-Z\beta\big\|_2^2
		+\lambda\sum_{j=1}^p\|\beta_j\|_2\right\},
		\qquad \text{s.t.}\qquad \mathbf{1}_p^{\top}\beta=0.
	\end{equation}
	where $\beta_{0}$ denotes the intercept, $Z_c$ represents additional non-compositional covariates, and $\beta_c$ is the corresponding coefficient vector. Here, $\beta$ is concatenated by subvectors $\{\beta_j\}_{j=1}^p$, where each $\beta_j$ collects the coefficients corresponding to the temporal basis functions of the $j$-th microbial taxon. The group Lasso penalty $\sum_{j=1}^p \|\beta_j\|_2$ facilitates the selection of the entire temporal trajectory for each taxon.
	It is also important to note that both methods extend compositional regression by incorporating taxonomic or temporal structures, but their performance under squared-error loss can still decline in the presence of outliers.

	\subsection{Robust Log-Contrast Regression }
	
	In real microbiome studies, outliers may arise from contamination, measurement error, or other technical artifacts in the sequencing process. 
	These outliers can seriously distort the standard estimators, since most existing compositional regression methods are built on the squared-error loss and are highly sensitive to extreme values. 
To address this issue, \cite{Mishra2022RobustCompositional} proposed a
robust regression framework for compositional covariates, referred to as
RobRegCC. The method augments the log-contrast regression model with a
mean-shift vector $\gamma=(\gamma_1,\ldots,\gamma_n)^\top$ to account for
potentially corrupted observations in the response. Restricting the notation
to compositional covariates, its general regularized formulation can be written as
\begin{equation}\label{eq:RobRegCC}
	\min_{\gamma\in\mathbb{R}^n,\,\beta\in\mathbb{R}^p}
	\left\{\frac{1}{2n}\left\|y-Z\beta-\sqrt{n}\,\gamma\right\|_2^2
	+P_{\lambda_1}^{(1)}(\gamma)+P_{\lambda_2}^{(2)}(\beta)\right\},
	\quad \text{s.t.}\quad C^{\top}\beta=0.
\end{equation}
Here, $P_{\lambda_1}^{(1)}(\gamma)$ and
$P_{\lambda_2}^{(2)}(\beta)$ denote sparsity-inducing penalties on the
mean-shift vector and the regression coefficients, respectively.
RobRegCC considers several penalty specifications, including the
hard-ridge penalty \citep{SheOwen2011Outlier}, the elastic-net penalty
\citep{ZouHastie2005ElasticNet}, and an adaptive elastic-net penalty
motivated by the adaptive Lasso \citep{Zou2006AdaptiveLasso}, leading to
a family of robust estimators.
It is worth noting that the $\ell_1$-penalized mean-shift formulation
has a well-known connection to Huber M-estimation. In the present work, we focus on
the direct constrained Huber formulation and its statistical and
computational treatment under general linear compositional constraints.
	Subsequently, \cite{monti2021sparse} introduced a robust sparse approach for compositional data, called Robust ZeroSum regression, which integrates trimming within the squared-loss function and a Lasso penalty. The estimator is given by
	\begin{equation}\label{eq:RZS}
		\min_{\beta\in\mathbb{R}^{p}}
		\left\{\frac{1}{2n}\sum_{i=1}^{n}w_i(y_i-z_i^{\top}\beta)^2
		+w\lambda\|\beta\|_1\right\},
		\qquad \text{s.t.}\qquad \mathbf{1}_p^{\top}\beta=0.
	\end{equation}
	where $w_i \in \{0,1\}$ indicates whether observation $i$ is retained after trimming and $w=\sum_{i=1}^n w_i$ denotes the proportion of retained observations.
	Finally, it is important to note that there are other related approaches, such as the one proposed by \cite{Srinivasan2021CompositionalKnockoff}, who introduced a compositional screening procedure (CSP) that adapts best-subset selection to the log-contrast regression model, resulting in an $\ell_0$-constrained squared-error loss function. This approach was later improved by \cite{Monti2024RobustKnockoff}, who developed a robust version of the CSP   that simultaneously achieves variable screening and robustness against outliers. More recently, \cite{monti2025robust} developed a robust
	multivariate regression framework for high-dimensional compositional
	microbiome data. Their approach accommodates multiple correlated responses
	and combines robust estimation with variable selection and false discovery
	rate control, thereby improving the stability of feature selection in the
	presence of outliers. \cite{zeng2026penalized} proposed a penalized
	composite quantile regression framework for compositional data under linear
	constraints.
	
Despite these advances, there remains interest in directly formulating
robust log-contrast regression through a robust loss function, without
introducing an explicit mean-shift vector for atypical observations.
Such a formulation provides a convex estimation problem while retaining
the linear structural constraints induced by compositional and
subcompositional covariates, and also leads to a distinct computational
treatment. Motivated by this perspective, we consider a constrained
log-contrast regression model based on the Huber loss
\citep{Huber1964RobustLocation}, defined as follows:
	
	\begin{equation}\label{eq:huber_model}
		\hat\beta=\arg\min_{\beta\in\mathbb{R}^{p}}
		\left\{\frac{1}{n}\sum_{i=1}^{n}h_{\tau}(y_i-z_i^\top\beta)
		+\lambda\|\beta\|_1\right\},
		\qquad \text{s.t.}\qquad C^{\top}\beta=0.
	\end{equation}
	where  $C \in \mathbb{R}^{p \times m}$ defined in \eqref{defC} is the coefficient matrix encoding the linear structure associated with compositional predictors.
	The univariate Huber loss \citep{Huber1964RobustLocation} is defined by
	\begin{equation}\label{eq:huber_def}
		h_{\tau}(u)	=	
		\begin{cases}
			\frac{1}{2}u^{2}, & |u| < \tau, \\[4pt]
			\tau|u| - \frac{1}{2}\tau^{2}, & |u| \geq \tau,
		\end{cases}
	\end{equation}
	where $\tau > 0$ is a tuning parameter controlling the trade-off between efficiency and robustness. Compared with the squared-error loss, the Huber loss is less sensitive to extreme observations while still retaining good efficiency for moderate residuals.
	The method that combines the Huber loss function and Lasso penalty for robust estimation and variable selection in standard linear regression has been studied, for examples, \cite{Sun2020Adaptive, LambertLacroix2011Robust}. But, the approach that adds the hierarchical structure of microbial taxa into log-contrast regression \eqref{eq:huber_model} has not been studied yet.
	
	\section{sGS-ADMM for Robust Log-contrast Regression}
	\label{sec:framework}
	In this section, we propose a computational algorithm to implement the estimation method defined in \eqref{eq:huber_model}, noting that the constraints in the minimization problem may introduce some computational challenges.
	To fully exploit the underlying structures, we employ an efficient ADMM based on the symmetric
	Gauss--Seidel iteration, which is both computationally efficient and guarantees convergence.
	
	\subsection{Preliminaries on Convex Optimization}\label{subpre}
	
	In this subsection, we review several fundamental concepts and preliminary results that are essential for the developments in the subsequent sections. 
	Let $\mathbb{R}^p$ be a finite-dimensional Euclidean space endowed with an inner product $\langle \cdot,\cdot \rangle $ and its induced norm $\|\cdot\|_2$, respectively.
	The inner product induced by a self-adjoint positive semidefinite  operator $\mathcal{S}$ is defined as  $\left \langle x,y \right \rangle _{\mathcal{S}}=\left \langle \mathcal{S}x,y \right \rangle $, and the corresponding norm is $\|x\|_{\mathcal{S}}^{2} := \langle x, x \rangle_{\mathcal{S}} = x^{\top} \mathcal{S} x$.
	Let $f:\mathbb{R}^p\to(-\infty,+\infty]$ be a closed proper convex function.
	The Fenchel conjugate function of $f$, denoted as $f^{*}$, is 
	\[
		f^{*}(y)
		:=\sup_{x\in\mathbb{R}^p}
		\big\{\langle y,x\rangle-f(x)\big\}
		=-\inf_{x\in\mathbb{R}^p}
		\big\{f(x)-\langle y,x\rangle\big\}.
	\]
	For a nonempty closed convex set $\mathcal{C}$, the symbol $\delta_{\mathcal{C}}(x)$ denotes 
	the indicator function over $\mathcal{C}$, that is, $\delta_{\mathcal{C}}(x)=0$ if $x \in \mathcal{C}$, and $+\infty$ otherwise. 
	For any $\ell_p$-norm function, it is not difficult to deduce that its Fenchel conjugate satisfies
	$\|x\|_{p}^{*}=\delta_{\mathbb{B}_{q}^{(1)}}(x)$ with $\mathbb{B}_{q}^{(1)}:=\{\, x \mid \|x\|_{q} \le 1 \,\}$,
	where $1/p + 1/q = 1$. 
	The subdifferential of $f$ at $x \in \mathrm{dom}(f)$ is defined by
	\[
		\partial f(x)
		:=\bigl\{\,g\in\mathbb{R}^p\ \big|\ 
		f(z)\ge f(x)+\langle g,z-x\rangle,
		\forall z\in\mathbb{R}^p
		\bigr\}.
	\]
	Any vector $g\in\partial f(x)$ is called a \emph{subgradient} of $f$ at $x$ \citep{Rockafellar1970ConvexAnalysis}. 
	When $f$ is differentiable at $x$, the subdifferential reduces to a singleton and $\partial f(x)=\{\nabla f(x)\}$.
	For the univariate Huber loss function defined in ~\eqref{eq:huber_def}, it is easy to deduce that its Fenchel conjugate function is $h_\tau^*(x)=\frac12x^2+\delta_{\mathbb{B}^{(\tau)}}(x)$, and its subdifferential has the following piecewise form:
	\[
	\partial h_{\tau}(x)=
	\begin{cases}
		\{x\}, & |x|< \tau,\\[4pt]
		\{\tau\,\mathrm{sign}(x)\}, & |x|\geq \tau,
	\end{cases}
	\]
	where $\mathrm{sign}(x)=1$ if $x>0$, $\mathrm{sign}(x)=-1$ if $x<0$, and $\mathrm{sign}(0)=0$ otherwise.
	
	The Moreau-Yosida regularization \citep{Moreau1962ProximalPoint} of a closed proper convex function $f$ at $x \in \mathbb{R}^p$ is defined as 
	\begin{equation}\label{eq:proximal mapping}
		\varphi _{f}(x) = \min\limits_{y \in \mathbb{R}^p} \left\{ f(y) + \frac{1}{2} \| y - x \|^{2} \right\}.
	\end{equation}
	For any $x \in \mathbb{R}^p$, problem~\eqref{eq:proximal mapping} has a unique optimal solution, which is called the proximal point of $x$ associated with $f$, or the proximal mapping operator, i.e.,
	\[
	\operatorname{prox}_{f}(x) = \arg\min\limits_{y \in \mathbb{R}^p} \left\{ f(y) + \frac{1}{2} \| y - x \|^{2} \right\}.
	\]
	In particular, the proximal point of $x$ associated with an indicator function 
	$\delta_{C}(x)$ reduces to the metric projection of $x$ onto a convex closed set $C$, i.e.,
	\begin{equation}\label{eq:prox_indicator}
		\begin{aligned}
			\operatorname{prox}_{\delta_{C}}(x)
			&= \arg\min_{y \in \mathbb{R}^p}
			\left\{ \delta_{C}(y) + \frac{1}{2}\|y - x\|^{2} \right\} \\
			&= \arg\min_{y \in C}
			\left\{ \frac{1}{2}\|y - x\|^{2} \right\} \\
			&= \Pi_{C}(x).
		\end{aligned}
	\end{equation}
	The relationship between the proximal point associated with a function and its conjugate is described as Moreau\textquotesingle s identity [\cite{Rockafellar1970ConvexAnalysis}, Theorem 35.1], i.e.,
	\begin{equation}\label{eq:moluodengshi}
		\operatorname{prox}_{tf}(x)+t\operatorname{prox}_{f^{*}/t}(x/t)=x,
	\end{equation}
	where $t>0$ is a positive scalar.
	The proximal mapping for some norm functions can be computed flexibly utilizing the Moreau\textquotesingle s  identity, which is essential in the subsequent analysis.
	As an example, consider the function $f(x)=\mu \|x\|_{1}$ with $\mu > 0$. Its Fenchel conjugate is given by $f^{*}(x^{*})=\delta _{\mathbb{B}_{\infty }^{(\mu)}}(x^{*})$ with 
	${\mathbb{B}_{\infty }^{(\mu)}}(x^{*}):=\big \{ {x^{*} \in \mathbb{R}^{p} \mid ||x^{*}||_{\infty }}\le \mu  \big \} $. Moreover, the proximal mapping operator of $f$ can be expressed as
	\begin{equation}
		\begin{aligned}
			\operatorname{prox}_{f}(x^{*})
			&=x^{*}-\Pi_{\mathbb{B}_{\infty }^{(\mu)}}(x^{*}),\\
			\big(\Pi_{\mathbb{B}_{\infty }^{(\mu)}}(x^{*})\big)_{i}
			&=
			\begin{cases}
				x_{i}^{*}, & |x_{i}^{*}|\le \mu,\\[4pt]
				\mathrm{sign}(x_{i}^{*})\mu, & |x_{i}^{*}|>\mu.
			\end{cases}
		\end{aligned}
	\end{equation}
	\subsection{sGS-ADMM Algorithm Based on Dual Formulation}
	This section focuses on deriving the dual formulation of the model in \eqref{eq:huber_model} and subsequently designing the sGS-ADMM algorithm.
	For convenience, we define the multivariate Huber function as
	\begin{equation}\label{eq:multi_huber}
		\mathcal{H}_{\tau}(x) = \frac{1}{n} \sum_{i=1}^{n} h_{\tau}(x_{i}),
	\end{equation}
	where $h_{\tau}(\cdot)$ is the univariate Huber loss function given by \eqref{eq:huber_def}.
	It follows from the convexity of $h_{\tau}(\cdot)$ that the multivariate Huber function $\mathcal{H}_{\tau}(\cdot)$ is also convex.
	Using this notation and defining $r := y - Z\beta$, then we rewrite  \eqref{eq:huber_model} equivalently as follows:
	\begin{equation}\label{eq:constrained_estimation}
		\begin{aligned}
			\min_{r \in \mathbb{R}^{n}, \, \beta \in \mathbb{R}^{p}} 
			& \quad \mathcal{H}_{\tau}(r) + p(\boldsymbol{\beta}),  \\
			\text{s.t.} & \quad 
			r = y - Z\beta, \quad C^{\top}\beta = 0,
		\end{aligned}
	\end{equation}
	where $p(\beta)=\lambda\|\beta\|_1$.
	
	The Lagrangian function associated with ~\eqref{eq:constrained_estimation} is
	\[
		\mathcal{L}(\beta, r; u, w)
		={} \mathcal{H}_\tau(r) + p(\beta) \\
		-\left\langle u,r-y+Z\beta\right\rangle
		-\left\langle w,C^{\top}\beta\right\rangle .
	\]
	where $u \in \mathbb{R}^n$ is the multiplier for $r = y - Z{\beta}$,  
	and $w \in \mathbb{R}^m$ be the multiplier for $C^{\top}{\beta}=0$.
	The Lagrange dual function of ~\eqref{eq:constrained_estimation} is to minimize $\mathcal{L}({\beta}, r; u, w)$ with respect
	to $r$ and $\beta$, that is
	\[
	\begin{aligned}
		D(u,w)
		&:=\max_{u,w}\min_{r,\beta}
		\mathcal{L}(\beta,r;u,w)\\
		&=\max_{u,w}\Big\{
		-\mathcal{H}_{\tau}^{*}(u)
		-p^{*}\big(Z^{\top}u+Cw\big)
		+\left\langle u,y\right\rangle
		\Big\}.
	\end{aligned}
	\]
	where $\mathcal{H}_{\tau}^{*}(\cdot)$ and $p^{*}(\cdot)$ represent the Fenchel conjugate functions 
	of $\mathcal{H}_{\tau}(\cdot)$ and $p(\cdot)$, respectively.
	The Lagrangian dual problem of  \eqref{eq:constrained_estimation} is to maximize the dual function $D(u,w)$, which can equivalently be written as the following minimization problem:
	\begin{equation}\label{eq:dual_min_problem}
		\min_{u \in \mathbb{R}^{n},\,w \in \mathbb{R}^{m}}
		\left\{\mathcal{H}_{\tau}^{*}(u)+p^{*}\big(Z^{\top}u+Cw\big)
		-\langle u,y\rangle\right\}.
	\end{equation}
	Using the preliminary results in Subsection \ref{subpre}, it is straightforward to calculate the Fenchel conjugate function of $\mathcal{H}_{\tau}^{*}(\cdot)$, which is given by 
	$\mathcal{H}_{\tau}^{*}(u) = \frac{n}{2}\|u\|^{2} + \delta_{\mathbb{B}_{\infty}^{(\tau/n)}}(u)$.
	Let $v:=Z^{\top}u+Cw$, then ~\eqref{eq:dual_min_problem} reformulates as the following separable form:
	\begin{equation}\label{eq:dual_box}
		\begin{aligned}
			\min_{\substack{u\in\mathbb{R}^{n},\,v\in\mathbb{R}^{p},
					w\in\mathbb{R}^{m}}}\quad
			&\mathcal{H}_{\tau}^{*}(u)+p^{*}(v)-\langle u,y\rangle,\\
			\text{s.t.}\quad &v-Z^{\top}u-Cw=0.
		\end{aligned}
	\end{equation}
	Obviously, we see that this minimization involves three separable variables with a linear constraint, which makes it suitable for solving using ADMM.
	
	For our purpose of using ADMM, we write the augmented Lagrangian function associated with dual problem ~\eqref{eq:dual_box} as follows:
	\begin{equation}\label{eq:augmented}
		\mathcal{L}_{\sigma}(u,v,w;\beta)
		=\mathcal{H}_\tau^{*}(u)+p^{*}(v)-\langle u,y\rangle
		-\langle\beta,\,v-Z^{\top}u-Cw\rangle
		+\frac{\sigma}{2}\|v-Z^{\top}u-Cw\|^2.
	\end{equation}
	where $\beta \in \mathbb{R}^{p}$ is multiplier associated with the constraint. Clearly, problem~\eqref{eq:dual_box} consists of three separable variable blocks
	$u$, $v$, and $w$ linked by the linear constraint $v-Z^{\top}u-Cw=0$.
	Since the classical ADMM is established for convex optimization problems with a two-block separable structure, it is not directly applicable to problem \eqref{eq:dual_box}, see \cite{ChenHeYeYuan2016}.
	Therefore, additional techniques are required to handle the three-block coupling structure. To preserve the intrinsic structure of the problem while maintaining
	computational tractability, we group the variables into two blocks, namely $(w,v)$ and $u$, and apply the symmetric Gauss--Seidel decomposition theorem of
	\cite{Li2018BlockSymmetricGS} to the $(w,v)$-block.
	Consequently, the augmented Lagrangian minimization problem can be solved by alternately updating the two blocks, while performing a symmetric Gauss--Seidel sweep within the block $(w,v)$.
	More precisely, starting from an initial point $(u^{0}, v^{0}, w^{0}; \beta^{0})$, the algorithm generates a sequence $\{(u^{k}, v^{k}, w^{k}; \beta^{k})\}$
	according to the following iterative scheme:
	\begin{equation}\label{eq:diedai}
		\left\{
		\begin{aligned}
			w^{k+\frac12}
			&= \arg\min_{w \in \mathbb{R}^{m}}
			\Big\{\mathcal{L}_{\sigma}(u^{k},v^{k},w;\beta^{k})
			+\tfrac12\|w-w^{k}\|_{\mathcal{Q}}^{2}\Big\},\\[1mm]
			v^{k+1}
			&= \arg\min_{v \in \mathbb{R}^{p}}
			\Big\{\mathcal{L}_{\sigma}(u^{k},v,w^{k+\frac12};\beta^{k})
			+\tfrac12\|v-v^{k}\|_{\mathcal{T}}^{2}\Big\},\\[1mm]
			w^{k+1}
			&= \arg\min_{w \in \mathbb{R}^{m}}
			\Big\{\mathcal{L}_{\sigma}(u^{k},v^{k+1},w;\beta^{k})
			+\tfrac12\|w-w^{k}\|_{\mathcal{Q}}^{2}\Big\},\\[1mm]
			u^{k+1}
			&= \arg\min_{u \in \mathbb{R}^{n}}
			\Big\{\mathcal{L}_{\sigma}(u,v^{k+1},w^{k+1};\beta^{k})
			+\tfrac12\|u-u^{k}\|_{\mathcal{S}}^{2}\Big\},\\[1mm]
			\beta^{k+1}
			&= \beta^{k}-\sigma\rho
			\bigl(v^{k+1}-Z^{\top}u^{k+1}-Cw^{k+1}\bigr).
		\end{aligned}
		\right.
	\end{equation}
	where $\rho\in(0,(1+\sqrt5)/2)$ is the dual step-length parameter,
	and $\mathcal S$, $\mathcal T$, and $\mathcal Q$ are self-adjoint
	positive semidefinite linear operators.
	Note that the variables $(w,v)$ are updated according to the symmetric
	Gauss--Seidel order
	\[
	w^{k+\frac12}\rightarrow v^{k+1}\rightarrow w^{k+1},
	\]
	whereas $u$ is updated subsequently as the second block.
	Consequently, algorithm \eqref{eq:diedai} can be interpreted as a
	two-block semi-proximal ADMM with an sGS decomposition applied to the
	$(w,v)$-block.

	We now discuss how each subproblem in
	\eqref{eq:diedai} can be solved efficiently.
	In particular, the special structure of the dual model enables
	closed-form solutions for the $(w,v)$-updates, whereas the
	$u$-subproblem can be computed via a simple projection operation.

First, let $\mathcal{Q}:=\gamma I$ with $\gamma > 0$.
Then the subproblem for computing $w^{k+\frac12}$ is given by
\begin{equation}\label{eq:w_update}
	\begin{aligned}
		w^{k+\frac12}
		&= \arg\min_{w\in\mathbb{R}^{m}}
		\Big\{\mathcal{L}_{\sigma}(u^{k},v^{k},w;\beta^{k})
		+\tfrac12\|w-w^{k}\|_{\mathcal Q}^{2}\Big\}\\
		&= \arg\min_{w\in\mathbb{R}^{m}}
		\Big\{\langle\beta^{k},Cw\rangle
		+\tfrac{\sigma}{2}\|v^{k}-Z^{\top}u^{k}-Cw\|^{2}
		+\tfrac{\gamma}{2}\|w-w^{k}\|^{2}\Big\}\\
		&= (\sigma C^{\top}C+\gamma I)^{-1}
		\Big[\sigma C^{\top}(v^{k}-Z^{\top}u^{k})
		+C^{\top}\beta^{k}+\gamma w^{k}\Big].
	\end{aligned}
\end{equation}
where $\gamma > 0$ is chosen such that $\sigma C^{\top}C+\gamma I$ is invertible.

Next, by setting $\mathcal{T}:=0$, the $v$-subproblem admits a particularly
simple form:
\begin{equation}\label{eq:v_update}
	\begin{aligned}
		v^{k+1}
		&:= \arg\min_{v\in\mathbb{R}^{p}}
		\Big\{\mathcal{L}_{\sigma}(u^{k},v,w^{k+\frac12};\beta^{k})
		+\tfrac12\|v-v^{k}\|_{\mathcal T}^{2}\Big\}\\
		&= \arg\min_{v\in\mathbb{R}^{p}}
		\Big\{p^{*}(v)-\langle\beta^{k},v\rangle
		+\tfrac{\sigma}{2}\|v-Z^{\top}u^{k}-Cw^{k+\frac12}\|^{2}\Big\}\\
		&= \operatorname{prox}_{p^{*}}^{1/\sigma}
		\left(Z^{\top}u^{k}+\frac{\beta^{k}}{\sigma}+Cw^{k+\frac12}\right)\\
		&=
		\left(Z^{\top}u^{k}+\frac{\beta^{k}}{\sigma}+Cw^{k+\frac12}\right)
		-\sigma^{-1}\operatorname{Prox}_{\sigma p}
		\left(\sigma\bigl(Z^{\top}u^{k}+Cw^{k+\frac12}\bigr)+\beta^{k}\right).
	\end{aligned}
\end{equation}
where the last equality is from the Moreau\textquotesingle s identity reviewed in \eqref{eq:moluodengshi}.

Finally, by choosing
$
\mathcal{S}:=(\eta-n)I_n-\sigma ZZ^\top,
$
where $\eta>0$ is selected such that $\mathcal{S}$ is positive semidefinite,
the variable $u$ can be updated at iteration $k$ by solving
\begin{equation}\label{eq:u_update}
	\begin{aligned}
		u^{k+1}
		&=\arg\min_{u\in\mathbb{R}^{n}}\Big\{
		\mathcal{L}_{\sigma}(u,v^{k+1},w^{k+1};\beta^{k})
		+\tfrac{1}{2}\|u-u^{k}\|_{\mathcal S}^{2}\Big\}\\
		&=\arg\min_{u\in\mathbb{R}^{n}}\Big\{
		\tfrac{n}{2}\|u\|^{2}
		+\delta_{\mathbb{B}_{\infty}^{(\tau/n)}}(u)
		-\langle u,y\rangle
		+\langle\beta^{k},Z^{\top}u\rangle\\[-1mm]
		&\qquad
		+\tfrac{\sigma}{2}
		\|v^{k+1}-Z^{\top}u-Cw^{k+1}\|^{2}
		+\tfrac{1}{2}\|u-u^{k}\|_{\mathcal S}^{2}\Big\}\\
		&=\Pi_{\mathbb{B}_{\infty}^{(\tau/n)}}
		\Big(\eta^{-1}\big[
		y-Z\beta^{k}
		+\sigma Zv^{k+1}
		+\mathcal Su^{k}
		-\sigma ZCw^{k+1}
		\big]\Big).
	\end{aligned}
\end{equation}
where the last equality is from \eqref{eq:prox_indicator}.
	
	We are now ready to present the H-RegAd (Huber Regression with sGS-ADMM) algorithm for robust constrained log-contrast regression on high-dimensional compositional microbiome data.
\begin{algorithm}[H]
	\caption{H-RegAd}
	\label{alg:sGS-ADMM}
	\begin{algorithmic}[1]
		
		\STATE \textbf{Initialization:}
		Choose $(u^{0},v^{0},w^{0};\beta^{0})$.
		Choose $\eta>0$ such that
		$\mathcal{S}:=(\eta-n)I-\sigma ZZ^{\top}\succeq0$,
		$\mathcal{Q}:=\gamma I$, $\mathcal{T}:=0$, $\sigma>0$, and
		$\rho\in(0,(1+\sqrt{5})/2)$.
		
		\FOR{$k=0,1,2,\ldots$}
		
		\STATE Compute $w^{k+\frac12}$ according to \eqref{eq:w_update}.
		
		\STATE Compute $v^{k+1}$ according to \eqref{eq:v_update}.
		
		\STATE Compute $w^{k+1}$ analogously to \eqref{eq:w_update},
		with $v^{k}$ replaced by $v^{k+1}$.
		
		\STATE Compute $u^{k+1}$ according to \eqref{eq:u_update}.
		
		\STATE Compute
		\[
		\beta^{k+1}
		=
		\beta^{k}
		-\sigma\rho
		\big(
		v^{k+1}
		-Z^{\top}u^{k+1}
		-Cw^{k+1}
		\big).
		\]
		
		\ENDFOR
		
		\STATE \textbf{Stop when a stopping criterion is satisfied.}
		
	\end{algorithmic}
\end{algorithm}
	We note that the stopping criterion of the proposed algorithm is based on the residuals of the KKT conditions.
	For clarity of presentation, the detailed definition of the stopping rule is deferred to the following numerical experiments section.
	Furthermore, the update sequence $w \rightarrow v \rightarrow w$  is equivalent to solving a single $(w,v)$-block subproblem under the symmetric
	Gauss--Seidel decomposition theorem of \cite{Li2018BlockSymmetricGS}
	with appropriate proximal term.
	Therefore, the proposed algorithm falls into the framework of two-block semi-proximal ADMM,
	and its global convergence follows directly from the general convergence theory of semi-proximal ADMM of  \cite{fazel2013hankel}.
	To conclude this subsection, we state the convergence result of Algorithm~\ref{alg:sGS-ADMM}. Its proof follows directly from \cite[Theorem B.1]{fazel2013hankel} and is therefore omitted.
	\begin{theorem}
		Let $\{(w^k,v^k,u^k,\beta^k)\}$	be the sequence generated by Algorithm \ref{alg:sGS-ADMM}.	
		Then the sequence $\{(w^k,v^k,u^k,\beta^k)\}$ converges to a solution $(\hat w,\hat v,\hat u,\hat \beta)$ of the Karush-Kuhn-Tucker (KKT) system associated with problem
		\eqref{eq:dual_box}. Consequently, $(\hat w,\hat v,\hat u)$  is an optimal solution of the dual problem
		\eqref{eq:dual_box}, and $\hat \beta$ 
		is an optimal solution of the corresponding primal problem \eqref{eq:huber_model}.
	\end{theorem}
	
	\subsection{Tuning Parameter Selection via Cross-Validation}
	To implement Algorithm H-RegAd, we use $K$-fold cross-validation to select the tuning parameter $\lambda$ and the Huber threshold $\tau$. Specifically, the data are randomly divided into $K$ approximately equal-sized folds. 
	For each candidate pair $(\lambda,\tau)$, the proposed robust constrained log-contrast regression model employing Huber loss \eqref{eq:huber_model} is fitted on the training folds, and its prediction performance is evaluated on the corresponding validation fold.
	Moreover, to reduce scale effects, the covariates are standardized within each training fold to have zero mean and unit variance, and the same standardization is subsequently applied to the corresponding validation fold.
	
	In the following experiments, we consider a decreasing sequence of candidate values for the regularization parameter $\lambda$, ranging from $\lambda_{\max}=\|Z^{\top}y\|_{\infty}/n$ to $10^{-3}\lambda_{\max}$, together with a small candidate set for the Huber threshold $\tau$.
	For the $k$-th fold, let $\widehat{\beta}^{(\setminus k)}(\lambda,\tau)$ denote the estimator obtained from the data excluding fold $k$. We then define the cross-validation criterion by
	\[
	\mathrm{CV}(\lambda,\tau)
	=\frac{1}{K}\sum_{k=1}^{K}\frac{1}{|I_k|}
	\sum_{i\in I_k}h_{\tau}\!\left(
	y_i-z_i^{\top}\widehat{\beta}^{(\setminus k)}(\lambda,\tau)\right).
	\]
	where $I_k$ is the index set of observations in the $k$-th validation fold.
	Then, the selected tuning parameters are given by
	$(\widehat{\lambda},\widehat{\tau})=\arg\min_{\lambda,\tau}\mathrm{CV}(\lambda,\tau)$.
	After selecting the optimal pair $(\widehat{\lambda},\widehat{\tau})$,
	the H-RegAd method is refitted using the full dataset. An adaptive thresholding rule is then applied to the estimator
	$\widehat{\beta}$ obtained from \eqref{eq:huber_model} to identify the
	active set. Subsequently, an unpenalized Huber refitting step is carried
	out on the selected variables to alleviate the shrinkage bias induced by
	the Lasso penalty in \eqref{eq:huber_model}. 
	Finally, among the variables retained after thresholding, a BIC-type
	criterion \citep{Wang2007Tuning} may be employed to select the final sparse model.
	
\section{Non-asymptotic Error Analysis}\label{sec:analysis}
In this section, we establish non-asymptotic guarantees for the proposed estimator \eqref{eq:huber_model}. 
Building upon several auxiliary lemmas and some regularity conditions, we will derive finite-sample error bounds comparable to those of existing Huber loss estimators.
To introduce sparsity-related notation, let $\beta^*\in\mathbb{R}^{p}$ denote the true coefficient vector and assume that it is $s$-sparse, with support set $S:=\mathrm{supp}(\beta^*)$ and cardinality $s:=|S|$. Throughout this paper, we consider the high-dimensional setting where the true coefficient $\beta^*$ is sparse and
$
s<n<p.
$

Let $\hat{\beta}$ be the coefficient estimator obtained from model~\eqref{eq:huber_model}, and define the estimation error vector by
$
h := \hat{\beta} - \beta^*.
$
Let $S_h$ denote the index set corresponding to the $s$ largest (in magnitude) components of $h$, and let $S_h^c$ be its complement. Finally, for any vector
$u = (u_1, u_2, \dots, u_n)^\top \in \mathbb{R}^n$,
let $\psi_{\tau}(u_i) \in \partial h_{\tau}(u_i)$ denote a subgradient of the univariate Huber loss function $h_{\tau}(u_i)$ and define
$
\Psi_{\tau}(u) := \big( \psi_{\tau}(u_1), \psi_{\tau}(u_2), \dots, \psi_{\tau}(u_n) \big)^\top \in\mathbb{R}^n.
$

We assume without loss of generality that the columns of the constraint matrix $C\in\mathbb{R}^{p\times m}$ are orthonormal, i.e.,
$
C^\top C=I_m.
$
Define
\[
P_C:=CC^\top,
\qquad
\widetilde Z:=Z(I_p-P_C).
\]
Since $C^\top\hat\beta=C^\top\beta^*=0$, it follows that
$
P_C\hat\beta=P_C\beta^*=0.
$
The projected matrix $\widetilde Z$ is introduced only for the theoretical analysis and does not change model~\eqref{eq:huber_model} or the numerical algorithm.

To characterize the geometry of the design matrix underlying our analysis, we introduce two standard quantities from high-dimensional statistics: the restricted orthogonality constant (ROC) and the restricted isometry property (RIP) constants. These quantities measure, respectively, the near-orthogonality between sparse directions and the stability of $A$ on sparse vectors.

\begin{definition}\label{ROC}
	Let $A \in \mathbb{R}^{n\times p}$. The restricted
	orthogonality constant (ROC) of orders $k_1$ and $k_2$
	is defined as
	\[
	\theta_{k_1,k_2}
	\left(\frac{A}{\sqrt{n}}\right)
	=
	\sup_{\substack{
			\alpha_1,\alpha_2\in\mathbb{R}^{p}\\
			\|\alpha_1\|_0\le k_1,\;
			\|\alpha_2\|_0\le k_2\\
			\operatorname{supp}(\alpha_1)
			\cap
			\operatorname{supp}(\alpha_2)=\varnothing
	}}
	\frac{
		|\langle A\alpha_1,A\alpha_2\rangle|
	}{
		n\|\alpha_1\|_2\|\alpha_2\|_2
	}.
	\]
\end{definition}

\begin{definition}\label{RIP}
	Let $A \in \mathbb{R}^{n\times p}$. The restricted isometry property (RIP) constants of order $k$ associated with $A/\sqrt{n}$ are defined as
	\[
	\begin{aligned}
		\delta_k^{+}\left(\frac{A}{\sqrt{n}}\right)
		&=\sup_{\alpha\in\mathbb{R}^{p}}
		\left\{
		\frac{\|A\alpha\|_2^{2}}{n\|\alpha\|_2^{2}}
		\ \middle|\ 
		\alpha\text{ is }k\text{-sparse}
		\right\},
	\end{aligned}
	\]
	and
	\[
	\begin{aligned}
		\delta_k^{-}\left(\frac{A}{\sqrt{n}}\right)
		&=\inf_{\alpha\in\mathbb{R}^{p}}
		\left\{
		\frac{\|A\alpha\|_2^{2}}{n\|\alpha\|_2^{2}}
		\ \middle|\ 
		\alpha\text{ is }k\text{-sparse}
		\right\}.
	\end{aligned}
	\]
	where $\delta_k^{+}\big(A/\sqrt{n}\big)$ and $\delta_k^{-}\big(A/\sqrt{n}\big)$ are the upper and lower RIP constants.
\end{definition}

\noindent Equivalently, the RIP condition means that for any $k$-sparse vector $\alpha \in \mathbb{R}^p$, we have
\[
\delta_k^{-}\left(\frac{A}{\sqrt{n}}\right)
\le
\frac{\|A\alpha\|_2^{2}}{n\|\alpha\|_2^{2}}
\le
\delta_k^{+}\left(\frac{A}{\sqrt{n}}\right).
\]

We note that these quantities will be used to control sparse cross-terms in the subsequent analysis, which are essential for deriving the non-asymptotic error bound below. We now present two technical lemmas as building blocks of the analysis.

\begin{lem}\label{convex}
	Let $h_\tau$ be the univariate Huber loss function and let
	$\psi_\tau(u_i)\in \partial h_\tau(u_i)$ denote its subgradient.
	Then for any $u_i,u_j\in\mathbb{R}$, it holds that
	\[
	h_{\tau}(u_i-u_j)
	\ge
	h_{\tau}(u_i)-\psi_{\tau}(u_i)u_j.
	\]
\end{lem}

\begin{lem}\label{lem:W}
	Let $h=\hat{\beta}-\beta^*$ and $u=y-\widetilde Z\beta^*$. Then
	\[
	\Psi_\tau(u)-\Psi_\tau(u-\widetilde Zh)
	=
	\int_0^1 D\Psi_\tau(u-t\widetilde Zh)\,dt\cdot\widetilde Zh,
	\]
	where $\Psi_\tau(u)=(\psi_\tau(u_1),\dots,\psi_\tau(u_n))^\top$ and
	$D\Psi_\tau(\cdot)$ denotes the diagonal matrix of almost everywhere derivatives of $\psi_\tau(\cdot)$. Define
	\[
	W:=\int_0^1\mathrm{diag}\big(\psi_\tau'(u-t\widetilde Zh)\big)\,dt.
	\]
	Then
	\[
	\Psi_\tau(u)-\Psi_\tau(u-\widetilde Zh)=W\widetilde Zh,
	\]
	and consequently,
	\[
	\widetilde Z^\top W\widetilde Zh
	=
	\widetilde Z^\top\Psi_\tau(u)
	-
	\widetilde Z^\top\Psi_\tau(u-\widetilde Zh).
	\]
	Moreover, the diagonal entries of $W$ lie in $[0,1]$.
\end{lem}

We remark that Lemma~\ref{convex} follows directly from the subgradient inequality for convex functions \citep{Rockafellar1970ConvexAnalysis}.
Lemma~\ref{lem:W} follows from the monotonicity and the $1$-Lipschitz continuity of the Huber subgradient function $\psi_{\tau}(\cdot)$.
Both lemmas relate the objective difference to subgradient terms which serve as key technical ingredients for our subsequent error bounds analysis.

To derive finite-sample error bounds, we impose a set of regularity conditions that control the stochastic error and ensure the restricted identifiability of the model through suitable sparse design assumptions. Since problem~\eqref{eq:huber_model} is convex with linear constraints, the KKT conditions are necessary and sufficient. Hence, there exist a Lagrange multiplier $\mu\in\mathbb{R}^m$ and a subgradient vector $\xi\in\partial\|\hat\beta\|_1$ such that
\begin{equation}\label{eq:kkt_condition}
	0
	=
	-\frac{1}{n}Z^\top\Psi_\tau(y-Z\hat\beta)
	+\lambda\xi
	+C\mu,
	\qquad
	C^\top\hat\beta=0.
\end{equation}

\begin{assumption}\label{KKT Condition}
	Assume that there exists a constant $k_0\ge1$, independent of $p$,
	such that
	\[
	\|I_p-P_C\|_\infty\le k_0,
	\]
	where $\|\cdot\|_\infty$ denotes the induced matrix infinity norm.
\end{assumption}

For the standard zero-sum constraint
$C=p^{-1/2}\mathbf{1}_p$, we have
$P_C=p^{-1}\mathbf{1}_p\mathbf{1}_p^\top$, and hence
\[
\|I_p-P_C\|_\infty
=
2\left(1-\frac{1}{p}\right)<2.
\]
Therefore, Assumption~\ref{KKT Condition} holds with
$k_0=2$. The same bound applies to disjoint blockwise
zero-sum constraints after normalizing each block indicator.

Premultiplying the first KKT equation in \eqref{eq:kkt_condition} by $I_p-P_C$ and using $(I_p-P_C)C=0$, we obtain
\[
\frac1n\widetilde Z^\top
\Psi_\tau(y-\widetilde Z\hat\beta)
=
\lambda(I_p-P_C)\xi.
\]
Since $\|\xi\|_\infty\le1$, Assumption~\ref{KKT Condition} implies
\begin{equation}\label{eq:k0bound}
	\left\|
	\frac1n\widetilde Z^\top
	\Psi_\tau(y-\widetilde Z\hat\beta)
	\right\|_\infty
	\le
	k_0\lambda.
\end{equation}

\begin{assumption}\label{zaoshengxiang}
	Assume that there exists a constant $\kappa>1$ such that the noise term $\epsilon$ in \eqref{eq:constrained_reg} satisfies
	\begin{equation}\label{eq:noise_score}
		\left\|
		\frac1n\widetilde Z^\top\Psi_\tau(\epsilon)
		\right\|_\infty
		\le
		\frac{\lambda}{\kappa}.
	\end{equation}
\end{assumption}

For a fixed design, Assumption~\ref{zaoshengxiang} holds with high
probability for $\lambda$ of order
$\tau\sqrt{(\log p)/n}$ under independent errors satisfying
$\mathbb E\{\psi_\tau(\epsilon_i)\}=0$ and normalized columns of
$\widetilde Z$.

Let
\[
a_\kappa:=\frac{\kappa+1}{\kappa-1},
\qquad
\mathcal I_\tau:=\left\{i:|\epsilon_i|\le\tau/2\right\},
\]
and define
\[
\mathcal C(S)
:=
\left\{
\Delta\in\mathbb R^p:
C^\top\Delta=0,\quad
\|\Delta_{S^c}\|_1\le a_\kappa\|\Delta_S\|_1
\right\}.
\]

\begin{assumption}\label{RE}
	Assume that there exist constants $K_Z<\infty$, $\phi_I>0$,
	${\delta}^{-}_{2s}>0$, ${\delta}^{+}_{2s}<\infty$, and
	$\phi_0>0$ such that
	\[
	\max_{1\le i\le n}\|\widetilde z_i\|_2\le K_Z,
	\]
	\[
	\frac1n\sum_{i\in\mathcal I_\tau}
	(\widetilde z_i^\top\Delta)^2
	\ge
	\phi_I\|\Delta\|_2^2,
	\qquad
	\Delta\in\mathcal C(S),
	\]
	and, for every $2s$-sparse vector $\alpha$,
	\[
	\frac1n\sum_{i\in\mathcal I_\tau}
	(\widetilde z_i^\top\alpha)^2
	\ge
	{\delta}^{-}_{2s}\|\alpha\|_2^2,
	\qquad
	\frac1n\|\widetilde Z\alpha\|_2^2
	\le
	{\delta}^{+}_{2s}\|\alpha\|_2^2.
	\]
	Moreover,
	\[
	(3\kappa-1){\delta}^{-}_{2s}
	-
	(\kappa+1){\delta}^{+}_{2s}
	\ge
	4\kappa\phi_0,
	\]
	and
	\begin{equation}\label{eq:localization_condition}
		\frac{2(\kappa+1)\sqrt{s}\lambda}
		{\kappa\phi_I}
		<
		\frac{\tau}{2K_Z}.
	\end{equation}
\end{assumption}

Assumption~\ref{RE} is stated only in terms of the design, the true noise, and the model parameters, and therefore does not depend on the fitted quantities $\hat\beta$ or $W$. It requires that a sufficient fraction of observations remains in the quadratic region of the Huber loss and that the corresponding design retains adequate sparse restricted curvature.

\begin{lem}\label{lem:localization}
	Under Assumptions~\ref{zaoshengxiang} and \ref{RE},
	\[
	\|h\|_2<\frac{\tau}{2K_Z}.
	\]
\end{lem}

\begin{pf}
	Since both $\hat\beta$ and $\beta^*$ are feasible, the optimality of
	$\hat\beta$ gives
	\[
	\frac1n\sum_{i=1}^n
	h_\tau(\epsilon_i-\widetilde z_i^\top h)
	+\lambda\|\beta^*+h\|_1
	\le
	\frac1n\sum_{i=1}^n h_\tau(\epsilon_i)
	+\lambda\|\beta^*\|_1.
	\]
	Using Lemma~\ref{convex} and Assumption~\ref{zaoshengxiang}, we obtain
	$
	\|h_{S^c}\|_1
	\le
	a_\kappa\|h_S\|_1,
	$
	and hence $h\in\mathcal C(S)$.
	
	Suppose that $\|h\|_2\ge r_0$, where
	$
	r_0:= \tau  / (2K_Z).
	$
	Let
	\[
	t:=\frac{r_0}{\|h\|_2},
	\qquad
	\bar h:=th.
	\]
	By convexity of the objective function, the same basic inequality holds with $h$ replaced by $\bar h$. Hence $\bar h\in\mathcal C(S)$ and $\|\bar h\|_2=r_0$.
	Moreover,
	\[
	|\widetilde z_i^\top\bar h|
	\le
	K_Zr_0
	=
	\tau/2.
	\]
	For $i\in\mathcal I_\tau$, we also have $|\epsilon_i|\le\tau/2$. Thus the Huber loss is quadratic along the line segment between $\epsilon_i$ and
	$\epsilon_i-\widetilde z_i^\top\bar h$, and therefore
\[
\frac1n\sum_{i=1}^n\left[h_\tau(\epsilon_i-\widetilde z_i^\top\bar h)-h_\tau(\epsilon_i)+\psi_\tau(\epsilon_i)\widetilde z_i^\top\bar h\right]\ge\frac1{2n}\sum_{i\in\mathcal I_\tau}(\widetilde z_i^\top\bar h)^2\ge\frac{\phi_I}{2}\|\bar h\|_2^2.
\]
	On the other hand, the basic inequality and Assumption~\ref{zaoshengxiang} give
	\[
		\frac1n\sum_{i=1}^n
		\left[
		h_\tau(\epsilon_i-\widetilde z_i^\top\bar h)
		-h_\tau(\epsilon_i)
		+\psi_\tau(\epsilon_i)\widetilde z_i^\top\bar h
		\right]
		\le
		\frac{\kappa+1}{\kappa}
		\lambda\sqrt{s}\|\bar h\|_2.
	\]
	Hence
	\[
	\|\bar h\|_2
	\le
	\frac{2(\kappa+1)\sqrt{s}\lambda}
	{\kappa\phi_I},
	\]
	which contradicts \eqref{eq:localization_condition}. Therefore,
	$\|h\|_2<\tau/(2K_Z)$.
\end{pf}

\begin{lem}\label{lem:weighted_rip}
	Under Assumptions~\ref{zaoshengxiang} and \ref{RE},
	\[
	\delta_{2s}^{-}
	\left(
	\frac{W^{1/2}\widetilde Z}{\sqrt n}
	\right)
	\ge
	{\delta}^{-}_{2s},
	\qquad
	\delta_{2s}^{+}
	\left(
	\frac{W^{1/2}\widetilde Z}{\sqrt n}
	\right)
	\le
	{\delta}^{+}_{2s}.
	\]
	Consequently,
	\begin{equation}\label{eq:weighted_RIP_relation}
		(3\kappa-1)
		\delta_{2s}^{-}
		\left(
		\frac{W^{1/2}\widetilde Z}{\sqrt n}
		\right)
		-
		(\kappa+1)
		\delta_{2s}^{+}
		\left(
		\frac{W^{1/2}\widetilde Z}{\sqrt n}
		\right)
		\ge
		4\kappa\phi_0.
	\end{equation}
\end{lem}

\begin{pf}
	By Lemma~\ref{lem:localization},
	\[
	|\widetilde z_i^\top h|
	\le
	K_Z\|h\|_2
	<
	\tau/2.
	\]
	For every $i\in\mathcal I_\tau$ and $t\in[0,1]$,
	\[
	|\epsilon_i-t\widetilde z_i^\top h|
	\le
	\tau.
	\]
	Hence $W_{ii}=1$ for $i\in\mathcal I_\tau$, while
	$0\le W_{ii}\le1$ for all $i$.
	
	For every $2s$-sparse vector $\alpha$,
	\[
		\frac1n\|W^{1/2}\widetilde Z\alpha\|_2^2
		=
		\frac1n\sum_{i=1}^n
		W_{ii}(\widetilde z_i^\top\alpha)^2\\
		\ge
		\frac1n\sum_{i\in\mathcal I_\tau}
		(\widetilde z_i^\top\alpha)^2
		\ge
		{\delta}^{-}_{2s}\|\alpha\|_2^2,
	\]
	and
	\[
	\frac1n\|W^{1/2}\widetilde Z\alpha\|_2^2
	\le
	\frac1n\|\widetilde Z\alpha\|_2^2
	\le
	{\delta}^{+}_{2s}\|\alpha\|_2^2.
	\]
	The result follows from Assumption~\ref{RE}.
\end{pf}

\begin{rmk}\label{rmk:assumptions}
	Assumption~\ref{KKT Condition} depends only on the constraint matrix, while
	Assumption~\ref{zaoshengxiang} is a standard high-dimensional score condition.
	Assumption~\ref{RE} is imposed on the true-error inlier design rather than on the
	fitted matrix $W^{1/2}\widetilde Z$. Thus, the conditions used below do not
	require assumptions on the unknown fitted multiplier or the fitted weight matrix.
\end{rmk}

\begin{theorem}
	Let $\hat{\beta}$ be the estimated solution of model~\eqref{eq:huber_model},
	and let $\beta^*$ denote the true coefficient vector, assumed to be $s$-sparse.
	Suppose that Assumptions~\ref{KKT Condition}--\ref{RE} hold. Then
	\[
	\|\hat{\beta}-\beta^*\|_1
	\le
	\frac{s\lambda\big(k_0+1/\kappa\big)}{\phi_0}.
	\]
\end{theorem}

\begin{pf}
	Since both $\hat{\beta}$ and $\beta^*$ satisfy the log-contrast
	constraint in \eqref{eq:huber_model}, the optimality of
	$\hat{\beta}$ implies that
	\[
		\frac{1}{n}\sum_{i=1}^n
		h_\tau\!\left(y_i-z_i^\top\hat{\beta}\right)
		+\lambda\|\hat{\beta}\|_1\\
		\le
		\frac{1}{n}\sum_{i=1}^n
		h_\tau\!\left(y_i-z_i^\top\beta^*\right)
		+\lambda\|\beta^*\|_1.
	\]
	Using $Z\hat\beta=\widetilde Z\hat\beta$,
	$Z\beta^*=\widetilde Z\beta^*$, and
	$y_i=z_i^\top\beta^*+\epsilon_i$, we have
	\begin{equation}\label{hebingshizi}
		\frac{1}{n}\sum_{i=1}^n
		\Big[
		h_\tau(\epsilon_i-\widetilde z_i^\top h)
		-h_\tau(\epsilon_i)
		\Big]
		\le
		\lambda\Big(\|\beta^*\|_1-\|\hat\beta\|_1\Big).
	\end{equation}
	Notice that $S:=\mathrm{supp}(\beta^*)$. It holds that
	\[
	\begin{aligned}
		\|\beta^*\|_1-\|\hat{\beta}\|_1
		&=\|\beta^*_{S}\|_1-\|\hat{\beta}_{S}\|_1
		-\|\hat{\beta}_{S^c}\|_1\\
		&\le
		\|\beta^*_{S}-\hat{\beta}_{S}\|_1
		-\|h_{S^c}\|_1\\
		&\le
		\|h_{S_h}\|_1-\|h_{S_h^c}\|_1.
	\end{aligned}
	\]
	where $S_h$ is the index set of the largest $s$ entries of $h$ in absolute value, and $S_h^c$ is its complement.
	
	Using Lemma~\ref{convex} and Assumption~\ref{zaoshengxiang}, we have
\[
\begin{aligned}
	\frac{1}{n}\sum_{i=1}^n\Big[h_\tau(\epsilon_i-\widetilde z_i^\top h)-h_\tau(\epsilon_i)\Big]
	&\ge -\frac{1}{n}h^\top\widetilde Z^\top\Psi_\tau(\epsilon)\\
	&\ge -\frac{1}{n}\|\widetilde Z^\top\Psi_\tau(\epsilon)\|_\infty\|h\|_1
	\ge -\frac{\lambda}{\kappa}\Big(\|h_{S_h}\|_1+\|h_{S_h^c}\|_1\Big).
\end{aligned}
\]
	Combining this with \eqref{hebingshizi}, we obtain
	\[
	-\frac{1}{\kappa}
	\big(\|h_{S_h}\|_1+\|h_{S_h^c}\|_1\big)
	\le
	\|h_{S_h}\|_1-\|h_{S_h^c}\|_1.
	\]
	Therefore,
	\begin{equation}\label{eq:cone}
		\|h_{S_h^c}\|_1
		\le
		\frac{\kappa+1}{\kappa-1}\|h_{S_h}\|_1.
	\end{equation}
	
	By \eqref{eq:k0bound},
	\[
	\left\|
	\widetilde Z^\top
	\Psi_\tau(y-\widetilde Z\hat\beta)
	\right\|_\infty
	\le
	k_0n\lambda.
	\]
	Since
	\[
	y-\widetilde Z\hat\beta
	=
	\epsilon-\widetilde Zh,
	\]
	Lemma~\ref{lem:W} and Assumption~\ref{zaoshengxiang} imply
\[
\|\widetilde Z^\top W\widetilde Zh\|_\infty\le\|\widetilde Z^\top\Psi_\tau(\epsilon)\|_\infty+\|\widetilde Z^\top\Psi_\tau(\epsilon-\widetilde Zh)\|_\infty\le\left(\frac1\kappa+k_0\right)n\lambda.
\]
	
	Applying Lemma~5.1 in \cite{caililuntuidao}, we can get
	\[
	\begin{aligned}
		&\abs{
			\inner{
				W^{1/2}\widetilde Zh_{S_h^c}
			}{
				W^{1/2}\widetilde Zh_{S_h}
		}}\\
		&\le
		n\,\theta_{s,s}
		\left(
		\frac{W^{1/2}\widetilde Z}{\sqrt n}
		\right)
		\|h_{S_h}\|_2
		\max\left(
		\|h_{S_h^c}\|_\infty,\,
		\frac{\|h_{S_h^c}\|_1}{s}
		\right)\sqrt{s}\\
		&\le
		n\,\frac{\kappa+1}{\kappa-1}
		\theta_{s,s}
		\left(
		\frac{W^{1/2}\widetilde Z}{\sqrt n}
		\right)
		\|h_{S_h}\|_2^2.
	\end{aligned}
	\]
	where we used
	\[
	\|h_{S_h^c}\|_\infty
	\le
	\frac{\|h_{S_h}\|_1}{s},
	\qquad
	\frac{\|h_{S_h^c}\|_1}{s}
	\le
	\frac{\kappa+1}{\kappa-1}
	\frac{\|h_{S_h}\|_1}{s},
	\]
	together with
	$\|h_{S_h}\|_1\le\sqrt{s}\|h_{S_h}\|_2$.
	
	Then,
\[
\begin{aligned}
	\left(k_0+\frac{1}{\kappa}\right)n\lambda\|h_{S_h}\|_1
	&\ge \|\widetilde Z^\top W\widetilde Zh\|_\infty\|h_{S_h}\|_1
	\ge \inner{\widetilde Z^\top W\widetilde Zh}{h_{S_h}}\\
	&=\|W^{1/2}\widetilde Zh_{S_h}\|_2^2
	+\inner{W^{1/2}\widetilde Zh_{S_h^c}}{W^{1/2}\widetilde Zh_{S_h}}\\
	&\ge n\left[
	\delta^-_{2s}\left(\frac{W^{1/2}\widetilde Z}{\sqrt n}\right)
	-\frac{\kappa+1}{\kappa-1}
	\theta_{s,s}\left(\frac{W^{1/2}\widetilde Z}{\sqrt n}\right)
	\right]\|h_{S_h}\|_2^2\\
	&\ge \frac{n}{2(\kappa-1)}
	\left[
	(3\kappa-1)\delta^-_{2s}\left(\frac{W^{1/2}\widetilde Z}{\sqrt n}\right)
	-(\kappa+1)\delta^+_{2s}\left(\frac{W^{1/2}\widetilde Z}{\sqrt n}\right)
	\right]
	\frac{\|h_{S_h}\|_1^2}{s}.
\end{aligned}
\]
	where the last inequality follows from Lemma~1 in
	\cite{Kang2016InvalidIVsisVIVE}, which states that for any matrix $A$,
	\[
	\theta_{k_1,k_2}(A)
	\le
	\frac12
	\Bigl(
	\delta^+_{k_1+k_2}(A)
	-
	\delta^-_{k_1+k_2}(A)
	\Bigr),
	\]
	together with
	$\|h_{S_h}\|_1\le\sqrt{s}\|h_{S_h}\|_2$.
	
	Using Lemma~\ref{lem:weighted_rip}, we have
\[
(3\kappa-1)\delta^-_{2s}\left(\frac{W^{1/2}\widetilde Z}{\sqrt n}\right)-(\kappa+1)\delta^+_{2s}\left(\frac{W^{1/2}\widetilde Z}{\sqrt n}\right)\ge 4\kappa\phi_0.
\]
	Therefore,
	\[
	\|h_{S_h}\|_1
	\le
	s\,
	\frac{
		\left(k_0+\frac1\kappa\right)(\kappa-1)\lambda
	}{
		2\kappa\phi_0
	}.
	\]
	Finally, by \eqref{eq:cone},
	\[
	\|\hat{\beta}-\beta^*\|_1
	=
	\|h_{S_h}\|_1+\|h_{S_h^c}\|_1
	\le
	\frac{2\kappa}{\kappa-1}\|h_{S_h}\|_1.
	\]
	Combining this with the bound for $\|h_{S_h}\|_1$, we obtain
	\[
	\|\hat{\beta}-\beta^*\|_1
	\le
	\frac{s\lambda\bigl(k_0+1/\kappa\bigr)}{\phi_0}.
	\]
\end{pf}

	\section{Simulation Studies}
	\label{sec:simulation}
	\subsection{Stopping Criterion}
	In this section, we evaluate the performance of H-RegAd (Algorithm~\ref{alg:sGS-ADMM}) and compare it with state-of-the-art methods for subcomposition regression \citep{Shi2016MicrobiomeRegression} and robust log-contrast regression \citep{Mishra2022RobustCompositional}. All numerical experiments were conducted in R (version 4.4.3) on a PC equipped with an Intel i5-10200H 2.4 GHz processor and 16 GB of RAM. 
	Moreover, all iterative algorithms were terminated when the KKT residuals defined below are sufficiently small. 
	
	We now present the KKT conditions for the dual problem, which will be used to define the stopping criteria. From convex optimization theory, solving the dual problem is equivalent to finding a quadruple $(u,v,w,\beta)$ satisfying the following KKT system:
	\[
	\left\{
	\begin{array}{l}
		y - Z\beta \in \partial H_{\tau}^{*}(u), \\[2mm]
		\beta \in \partial p^{*}(v), \\[2mm]
		C^{\top}\beta = 0, \\[2mm]
		v - Z^{\top}u - Cw = 0.
	\end{array}
	\right.
	\]
	Based on the above KKT system, we define the following relative residuals:
	\[
	\begin{aligned}
		\eta_{1}
		&:=\frac{\left\|u-\Pi_{\mathbb{B}_{\infty}^{(\tau/n)}}
		\bigl(n^{-1}(y-Z\beta)\bigr)\right\|}{1+\|u\|},
		\qquad
		\eta_{2}
		:=\frac{\left\|v-\Pi_{\mathbb{B}_{\infty}^{(\lambda)}}(v+\beta)
		\right\|}{1+\|v\|},\\[1mm]
		\eta_{3}
		&:=\frac{\|C^{\top}\beta\|}{1+\|\beta\|},
		\qquad
		\eta_{4}
		:=\frac{\|v-Z^{\top}u-Cw\|}
		{1+\|v\|+\|Z^{\top}u\|+\|Cw\|}.
	\end{aligned}
	\]
	Hence, the KKT residual for problem (\ref{eq:dual_box}) takes the following form, which is a natural choice for the stopping criterion of Algorithm \ref{alg:sGS-ADMM}:
	\[
	\eta_{d}=\max\big\{\eta_{1},\eta_{2},\eta_{3},\eta_{4}\big\}.
	\]
	Therefore, a natural stopping criterion for the proposed dual Algorithm \ref{alg:sGS-ADMM} is
	$\eta_{d}\le 10^{-4}$, with the maximum number of iterations set at 2000. We stop the main iteration when the improvement in the objective function of (\ref{eq:dual_box}) is below a threshold.

	\subsection{Simulation Settings}
	Following the simulation settings of \cite{Mishra2022RobustCompositional}, we generate positive-valued data
	$\mathbf{W}=(\mathbf{w}_1,\ldots,\mathbf{w}_n)^\top=(w_{ij}) \in \mathbb{R}^{n\times p}$,
	where $\mathbf{w}_i$ are independent realizations of a multivariate lognormal distribution, i.e.,
	$\log(\mathbf{w}_i)\sim \mathcal{N}(\boldsymbol{\mu},\boldsymbol{\Sigma})$.
	Here, the mean vector
	$\boldsymbol{\mu}=\big(\log(p/2)\mathbf{1}_5^\top,\;\mathbf{0}_{p-5}^\top\big)^\top\in\mathbb{R}^p$
	and the covariance matrix $\boldsymbol{\Sigma}\in\mathbb{R}^{p\times p}$ satisfies
	$\Sigma_{ij}=0.5^{|i-j|}$.
	Then, we perform total-sum normalization on $\mathbf{W}$ and apply a log-transformation to the resulting compositions to yield the covariate matrix
	$Z=(z_{ij})_{n\times p}$ with
	\[
	z_{ij}
	=
	\log\left(\frac{w_{ij}}{\sum_{l=1}^p w_{il}}\right).
	\]
	For simplicity, we include a single non-compositional covariate in the form of an intercept
	${N}=\mathbf{1}_n$.
	According to the generative model in \eqref{eq:huber_model}, we define
	${X}:=[{N}\ Z]$
	and set
	$
		\beta^*
		=\Bigl(
		\beta_0,\;1,\;-0.8,\;0.4,\;0,\;0,\;-0.6,\;0,\;0,0,\;0,\;-1.5,\;0,\\
		\;1.2,\;0,\;0,\;0.3,\;
		\mathbf{0}_{p-16}
		\Bigr).
	$
	with $\beta_0=0.5$ as the choice in \cite{Shi2016MicrobiomeRegression}.
	To enforce subcompositional coherence, we consider the subcomposition constraint matrix $C$ with $m=4$ subgroups of the form
	\begin{equation}\label{eq:C}
		C^{\top}
		=\begin{bmatrix}\mathbf{0}_{4\times 1} & \bar C^{\top}\end{bmatrix}_{4\times(p+1)},
		\qquad
		\bar C^{\top}
		=\begin{bmatrix}
			\mathbf{1}_{p_1}^{\top} & 0 & 0 & 0\\
			0 & \mathbf{1}_{p_2}^{\top} & 0 & 0\\
			0 & 0 & \mathbf{1}_{p_3}^{\top} & 0\\
			0 & 0 & 0 & \mathbf{1}_{p_4}^{\top}
		\end{bmatrix}.
	\end{equation}
	where $p= {\textstyle \sum_{i=1}^{4}} p_{i} $ with $p_i=\varsigma_{i+1}-\varsigma_i$ and $\varsigma:=(0,10,16,20,p)$, and index sets $\mathbb{A}_i=\{\varsigma_i+1,\varsigma_i+2,\ldots,\varsigma_{i+1}\}$.
	Using these settings, we first generate the clean response vector $y$ with additive noise ${\epsilon}$, where $\epsilon_i \sim \mathcal{N}(0,\sigma^2)$.
	The error standard deviation was chosen as
	\[
	\sigma = \frac{\|Z\beta^*\|}{\sqrt{n}\times \mathrm{SNR}},
	\]
	where the signal-to-noise ratio was fixed at $\mathrm{SNR}=3$. In the simulation study, we consider a fixed sample size $n=200$ and a high-dimensional setting with $p=300$.
	
	To assess the ability of H-RegAd to identify outliers in the response, we considered several contamination scenarios. Specifically, the proportion of outliers was set to $\mathcal{O} \in \{0\%,10\%, 20\%, 30\%, 40\%\}$ and the number of outliers was set $O=\mathcal{O}n$. Moderate outliers were generated by adding a shift of $6\sigma$ to the corresponding response entries, i.e., $y_i\leftarrow y_i + s_f$ with $s_f=6\sigma$ for $i \in \mathcal{O}[n]$.
	We also consider a more challenging setting in which half of the $O$ response outliers were additionally set as leverage points. The leverage points in $\mathbf{Z}$, denoted by $L$ in the simulation settings, were generated by modifying the corresponding observations in the count matrix $\mathbf{W}$. Specifically, the first $O/2$ observations of $\mathbf{W}$ were replaced by leveraged observations. For each leveraged observation, one taxon within each subgroup was inflated to a relatively large value, while the remaining taxa in the same subgroup were reduced to small values. For each of the $m$ subgroups, we extracted the corresponding columns of $\mathbf{W}$. After adding the constant $4$ to the first column, we arranged this column in descending order and the remaining columns in ascending order. The first $O/2$ observations in the resulting matrix were then treated as leverage points. Here, $L=0$ indicates that no leverage points were introduced, whereas $L=1$ indicates that half of the response outliers were also contaminated in the design matrix. Each experimental setting was replicated $100$ times.

	We assess the estimation accuracy of the sparse regression coefficients using the scaled estimation error defined as
	\[
	\mathrm{Er}({\hat\beta}) =
	\frac{100\,\lVert \hat{{\beta}}-{\beta}^*\rVert_2}{p},
	\]
	where $\hat{{\beta}}$ and ${\beta}^*$ denote the estimated and true coefficient vectors, respectively.
	Furthermore, to assess the ability of each method to detect outliers, we report the numbers of false positives (FP) and false negatives (FN), obtained by comparing the estimated outlier set with the true outlier set. 
	The overall outlier detection error is measured by the Hamming distance defined as
	$\mathrm{HM}=\mathrm{FN}+\mathrm{FP}$.
	Finally, to further evaluate the predictive performance of each method, we report the mean absolute error (MAE), accuracy (Acc), and precision (Prec), defined as, respectively,
\[
\mathrm{MAE}:=\frac{1}{n}\sum_{i=1}^{n}\lvert y_i-\hat{y}_i\rvert,\qquad
\mathrm{Acc}:=\frac{\mathrm{TP}+\mathrm{TN}}{\mathrm{TP}+\mathrm{TN}+\mathrm{FP}+\mathrm{FN}},\qquad
\mathrm{Prec}:=\frac{\mathrm{TP}}{\mathrm{TP}+\mathrm{FP}}.
\]
	where $\hat{y}_i$ is the predicted response for the $i$-th observation, TP and TN denote the numbers of true positives and true negatives, respectively. Specifically, MAE measures the average absolute prediction error, Acc quantifies the proportion of correctly classified samples, and Prec measures the proportion of true positives among all predicted positive samples.

	\subsection{Simulation Results}
	In this subsection, we aim to report the simulation results of H-RegAd under the outlier setting with shift $s_f=6\sigma$. 
	Specifically, the results under the simple constraint $C=\mathbf{1}_p$ are reported in Table~\ref{tab:C1}, while the results under the block-structured constraint matrix in \eqref{eq:C} are reported in Table \ref{tab:C2}. 
	In addition, we fix all constraint hyperparameters and conduct comparative analyses between H-RegAd and the RobRegCC method proposed in \cite{Mishra2022RobustCompositional}. We assess three penalty configurations for RobRegCC, including hard-ridge (RobRegCC-H), adaptive (RobRegCC-A), and elastic-net (RobRegCC-E), as well as the standard log-contrast model RobRegCC-NR without mean shift adjustment. 
	We note that the R package for RobRegCC is available at \url{https://github.com/amishra-stats/robregcc}, which also includes the RobRegCC-NR method proposed in \cite{Shi2016MicrobiomeRegression}.
	In terms of initialization, RobRegCC-H and RobRegCC-A employ a PSC-based initialization procedure. Meanwhile, the elastic-net variant RobRegCC-E is implemented by letting $\mathbf{w}=\mathbf{1}_{n+p}$, and its corresponding adaptive penalty function is denoted by  RobRegCC-A.

	From Table \ref{tab:C1}, where we impose the zero-sum log-contrast constraint $\mathbf{1}_{p}^{\mathsf T}\beta=0$, we observe that H-RegAd achieves the best overall performance under moderate contamination. We find that H-RegAd yields  smaller estimation error $\mathrm{Er}(\hat{\beta})$ than all RobRegCC variants, and we verify that this performance advantage is consistently retained as the contamination level increases.
	In terms of outlier detection, H-RegAd consistently achieves the smallest Hamming distance while maintaining precision close to
	$1$.  Overall, these results demonstrate that H-RegAd achieves the best trade-off between coefficient estimation accuracy and outlier identification performance.
	
	From Table~\ref{tab:C2}, where the subcomposition constraint matrix $C$ is formed as in \eqref{eq:C}, we observe that H-RegAd again demonstrates the best overall performance under moderate contamination.
	Regarding coefficient estimation, H-RegAd consistently attains smaller estimation errors $\mathrm{Er}(\hat{\beta})$ than the competing methods. 
	For outlier detection, it also achieves superior performance, typically yielding the smallest Hamming distance together with high detection accuracy. 
	Moreover, in contaminated scenarios, its precision is consistently high and even reaches $1$ in several settings. 
	Although H-RegAd does not always produce the lowest MAE, its predictive performance remains highly competitive across nearly all cases.
	Overall, these results further confirm that H-RegAd provides the most reliable and well-balanced performance in both coefficient estimation and outlier identification under the subcomposition constraint.
	
	\begin{sidewaystable}[!p]
		\centering
		\caption{Simulation performance under moderate contamination with the zero-sum constraint
			$(n=200,\ p=300)$.}
		\label{tab:C1}
		
		\fontsize{7}{8}\selectfont
		\setlength{\tabcolsep}{0.25pt}
		\renewcommand{\arraystretch}{1.5}
		\setlength{\extrarowheight}{1.2pt}
		
		\begingroup
		\begin{tabular}{ccc|ccccc|ccccc|ccccc|ccccc|ccccc}
			\hline
			
			\multirow{2}{*}{$L$}
			& \multirow{2}{*}{$p$}
			& \multirow{2}{*}{$\mathcal{O}$}
			& \multicolumn{5}{c|}{\textbf{[RobRegCC-A]}}
			& \multicolumn{5}{c|}{\textbf{[RobRegCC-H]}}
			& \multicolumn{5}{c|}{\textbf{[RobRegCC-E]}}
			& \multicolumn{5}{c|}{\textbf{[RobRegCC-NR]}}
			& \multicolumn{5}{c}{\textbf{[H-RegAd]}} \\
			
			\cline{4-8}
			\cline{9-13}
			\cline{14-18}
			\cline{19-23}
			\cline{24-28}
			
			& & &
			Er & HM & MAE & Acc & Prec
			& Er & HM & MAE & Acc & Prec
			& Er & HM & MAE & Acc & Prec
			& Er & HM & MAE & Acc & Prec
			& Er & HM & MAE & Acc & Prec \\
			
			\hline
			
			0 & 300 & 0
			& 6.896e-3 & 7  & 1.218 & 0.977 & 0.941
			& 6.896e-3 & 7  & 1.218 & 0.977 & 0.941
			& 6.896e-3 & 7  & 1.218 & 0.977 & 0.941
			& 5.857e-3 & 12 & 1.157 & 0.960 & 0.895
			& \textbf{1.120e-3} & \textbf{3}
			& \textbf{1.016} & \textbf{0.990} & \textbf{0.950} \\
			
			0 & 300 & 10\%
			& 5.810e-3 & 10 & 1.617 & 0.967 & 0.912
			& 5.810e-3 & 10 & 1.617 & 0.967 & 0.912
			& 5.827e-3 & 12 & 1.626 & 0.960 & 0.907
			& 1.113e-2 & 8 & 1.796 & 0.953 & 0.899
			& \textbf{1.546e-3} & \textbf{4}
			& \textbf{1.701} & \textbf{0.987} & \textbf{1.000} \\
			
			0 & 300 & 20\%
			& 6.008e-3 & 11 & 1.886 & 0.963 & 0.908
			& 6.008e-3 & 11 & 1.886 & 0.963 & 0.908
			& 6.416e-3 & 8 & 1.920 & 0.970 & 0.915
			& 1.095e-2 & 10 & 1.983 & 0.955 & 0.906
			& \textbf{1.672e-3} & \textbf{4}
			& \textbf{1.863} & \textbf{0.987} & \textbf{1.000} \\
			
			0 & 300 & 30\%
			& 7.761e-3 & 9 & 2.202 & 0.970 & 0.937
			& 7.761e-3 & 9 & 2.202 & 0.970 & 0.937
			& 1.084e-2 & 10 & \textbf{2.072} & 0.967 & 0.925
			& 1.313e-2 & 11 & 2.469 & 0.923 & 0.884
			& \textbf{2.677e-3} & \textbf{2}
			& 2.210 & \textbf{0.993} & \textbf{1.000} \\
			
			0 & 300 & 40\%
			& 6.014e-3 & 7 & 2.911 & 0.977 & 0.911
			& 6.014e-3 & 7 & 2.911 & 0.977 & 0.911
			& 1.009e-2 & 8 & 2.818 & 0.965 & 0.918
			& 1.597e-2 & 9 & 3.099 & 0.963 & 0.900
			& \textbf{2.175e-3} & \textbf{4}
			& \textbf{2.890} & \textbf{0.987} & \textbf{1.000} \\
			
			\hline
			
			1 & 300 & 0
			& 3.667e-3 & 6 & 1.703 & 0.980 & 0.946
			& 3.667e-3 & 6 & 1.703 & 0.980 & 0.946
			& 3.667e-3 & 6 & 1.703 & 0.980 & 0.946
			& 4.480e-3 & 7 & 1.710 & 0.967 & 0.923
			& \textbf{3.544e-3} & \textbf{3}
			& \textbf{1.401} & \textbf{0.990} & \textbf{0.970} \\
			
			1 & 300 & 10\%
			& 4.994e-3 & 9 & 1.735 & 0.970 & 0.937
			& 4.994e-3 & 9 & 1.735 & 0.970 & 0.937
			& 4.994e-3 & 9 & 1.735 & 0.970 & 0.937
			& 6.566e-3 & 7 & 1.831 & 0.970 & 0.926
			& \textbf{6.203e-3} & \textbf{3}
			& \textbf{1.711} & \textbf{0.990} & \textbf{1.000} \\
			
			1 & 300 & 20\%
			& 6.462e-3 & 6 & 2.175 & 0.980 & 0.946
			& 6.462e-3 & 6 & 2.175 & 0.980 & 0.946
			& 6.680e-3 & 5 & 2.229 & 0.987 & 0.957
			& 6.347e-3 & 7 & 2.684 & 0.943 & 0.924
			& \textbf{1.263e-3} & \textbf{1}
			& \textbf{2.156} & \textbf{0.997} & \textbf{1.000} \\
			
			1 & 300 & 30\%
			& 5.380e-3 & 7 & 2.894 & 0.954 & 0.921
			& 5.380e-3 & 7 & 2.894 & 0.954 & 0.921
			& 5.380e-3 & 7 & 2.894 & 0.954 & 0.921
			& 1.163e-2 & 8 & 3.473 & 0.937 & 0.893
			& \textbf{4.553e-3} & \textbf{2}
			& \textbf{2.032} & \textbf{0.993} & \textbf{1.000} \\
			
			1 & 300 & 40\%
			& 7.365e-3 & 7 & 3.293 & 0.963 & 0.934
			& 7.365e-3 & 7 & 3.293 & 0.963 & 0.934
			& 7.365e-3 & 7 & 3.293 & 0.963 & 0.934
			& 1.828e-2 & 8 & 3.770 & 0.921 & 0.885
			& \textbf{1.788e-3} & \textbf{5}
			& \textbf{3.110} & \textbf{0.983} & \textbf{1.000} \\
			
			\hline
		\end{tabular}%
		\endgroup
	\end{sidewaystable}
	
	\begin{sidewaystable}[!p]
		\centering
		\caption{Simulation performance under moderate contamination with the subcomposition constraint
			$(n=200,\ p=300)$.}
		\label{tab:C2}
		
		\fontsize{7}{8}\selectfont
		\setlength{\tabcolsep}{0.25pt}
		\renewcommand{\arraystretch}{1.5}
		\setlength{\extrarowheight}{1.2pt}
		
		\begingroup
		\begin{tabular}{ccc|ccccc|ccccc|ccccc|ccccc|ccccc}
			\hline
			
			\multirow{2}{*}{$L$}
			& \multirow{2}{*}{$p$}
			& \multirow{2}{*}{$\mathcal{O}$}
			& \multicolumn{5}{c|}{\textbf{[RobRegCC-A]}}
			& \multicolumn{5}{c|}{\textbf{[RobRegCC-H]}}
			& \multicolumn{5}{c|}{\textbf{[RobRegCC-E]}}
			& \multicolumn{5}{c|}{\textbf{[RobRegCC-NR]}}
			& \multicolumn{5}{c}{\textbf{[H-RegAd]}} \\
			
			\cline{4-8}
			\cline{9-13}
			\cline{14-18}
			\cline{19-23}
			\cline{24-28}
			
			& & &
			Er & HM & MAE & Acc & Prec
			& Er & HM & MAE & Acc & Prec
			& Er & HM & MAE & Acc & Prec
			& Er & HM & MAE & Acc & Prec
			& Er & HM & MAE & Acc & Prec \\
			
			\hline
			
			0 & 300 & 0
			& 6.667e-3 & 5 & 1.184 & 0.923 & 0.899
			& 6.433e-3 & 5 & 1.184 & 0.923 & 0.899
			& 6.667e-3 & 5 & 1.184 & 0.923 & 0.899
			& 6.187e-3 & 5 & 1.140 & 0.953 & 0.919
			& \textbf{1.020e-3} & \textbf{3}
			& \textbf{0.997} & \textbf{0.990} & \textbf{0.986} \\
			
			0 & 300 & 10\%
			& 5.150e-3 & 8 & \textbf{1.598} & 0.943 & 0.913
			& 5.150e-3 & 8 & 1.702 & 0.943 & 0.913
			& 5.150e-3 & 8 & 1.702 & 0.943 & 0.913
			& 6.704e-3 & 5 & 1.742 & 0.936 & 0.899
			& \textbf{3.364e-3} & \textbf{2}
			& 1.602 & \textbf{0.993} & \textbf{1.000} \\
			
			0 & 300 & 20\%
			& 7.507e-3 & 8 & 1.849 & 0.936 & 0.908
			& 5.450e-3 & 8 & 1.846 & 0.936 & 0.908
			& 5.450e-3 & 8 & 1.846 & 0.936 & 0.908
			& 1.120e-2 & 5 & 2.035 & 0.912 & 0.887
			& \textbf{1.122e-3} & \textbf{3}
			& \textbf{1.842} & \textbf{0.990} & \textbf{0.986} \\
			
			0 & 300 & 30\%
			& 9.157e-3 & 10 & 2.264 & 0.967 & 0.933
			& 9.157e-3 & 9 & 2.173 & 0.967 & 0.933
			& 9.157e-3 & 9 & 2.173 & 0.967 & 0.933
			& 1.816e-2 & 9 & 2.422 & 0.946 & 0.915
			& \textbf{3.512e-3} & \textbf{2}
			& \textbf{2.094} & \textbf{0.987} & \textbf{0.968} \\
			
			0 & 300 & 40\%
			& 7.070e-3 & 4 & 2.825 & 0.977 & 0.945
			& 7.070e-3 & 5 & 2.835 & 0.977 & 0.945
			& 7.070e-3 & 5 & 2.835 & 0.977 & 0.945
			& 1.430e-2 & 7 & 3.062 & 0.917 & 0.879
			& \textbf{1.858e-3} & \textbf{4}
			& \textbf{2.029} & \textbf{0.993} & \textbf{1.000} \\
			
			\hline
			
			1 & 300 & 0
			& 4.589e-3 & 13 & 1.421 & 0.957 & 0.945
			& 4.589e-3 & 13 & 1.421 & 0.957 & 0.945
			& 5.145e-3 & 13 & 1.455 & 0.957 & 0.945
			& 1.872e-2 & 7 & 2.727 & 0.957 & 0.938
			& \textbf{1.911e-3} & \textbf{2}
			& \textbf{1.211} & \textbf{0.978} & \textbf{0.951} \\
			
			1 & 300 & 10\%
			& 8.152e-3 & 17 & 1.506 & 0.943 & 0.922
			& 8.152e-3 & 19 & 1.487 & 0.943 & 0.922
			& 8.152e-3 & 19 & 1.487 & 0.943 & 0.922
			& 6.376e-2 & 16 & 1.859 & 0.933 & 0.876
			& \textbf{5.822e-3} & \textbf{8}
			& \textbf{1.300} & \textbf{0.973} & \textbf{0.949} \\
			
			1 & 300 & 20\%
			& 3.842e-3 & 12 & 1.996 & 0.976 & 0.942
			& 3.842e-3 & 13 & 1.996 & 0.976 & 0.942
			& 3.842e-3 & 13 & 1.996 & 0.976 & 0.942
			& 1.713e-2 & 10 & 2.245 & 0.926 & 0.887
			& \textbf{1.409e-3} & \textbf{7}
			& \textbf{1.752} & \textbf{0.987} & \textbf{0.976} \\
			
			1 & 300 & 30\%
			& 6.302e-3 & 15 & 2.010 & 0.950 & 0.922
			& 6.253e-3 & 16 & 1.963 & 0.950 & 0.922
			& 6.253e-3 & 16 & \textbf{1.963} & 0.950 & 0.922
			& 1.628e-2 & 9 & 2.369 & 0.918 & 0.864
			& \textbf{1.553e-3} & \textbf{2}
			& 1.972 & \textbf{0.993} & \textbf{1.000} \\
			
			1 & 300 & 40\%
			& 7.404e-3 & 13 & 2.462 & 0.957 & 0.934
			& 7.404e-3 & 14 & 2.462 & 0.957 & 0.934
			& 7.404e-3 & 14 & 2.462 & 0.957 & 0.934
			& 1.636e-2 & 7 & 1.727 & 0.918 & 0.855
			& \textbf{3.865e-3} & \textbf{2}
			& \textbf{2.421} & \textbf{0.993} & \textbf{1.000} \\
			
			\hline
		\end{tabular}%
		\endgroup
	\end{sidewaystable}
	
	While Tables~\ref{tab:C1} and~\ref{tab:C2} report the average performance of each method, we also investigate the stability of all the methods across repeated simulations. Specifically, we present boxplots of HM, Er, and MAE under different outlier proportions for both the zero-sum and subcomposition constraints in Figure~\ref{fig:boxplots}.
	From Figure \ref{fig:boxplots}, we observe that HM, Er and MAE generally increase as the contamination proportion becomes larger across all methods, which indicates that heavier contamination poses greater challenges for both outlier identification and predictive performance. In contrast, H-RegAd consistently maintains relatively low HM and Er under both constraint structures, which indicates that H-RegAd achieves greater stability in outlier detection and coefficient recovery. 
	Overall, although MAE increases as contamination becomes more severe, H-RegAd remains competitive in prediction accuracy in most settings. Moreover, we observe that H-RegAd is consistently less sensitive to contaminated observations than the method RobRegCC-NR, which further confirms its robustness.

	
	\begin{figure}[!htbp]
		\centering
		\includegraphics[width=0.65\linewidth]{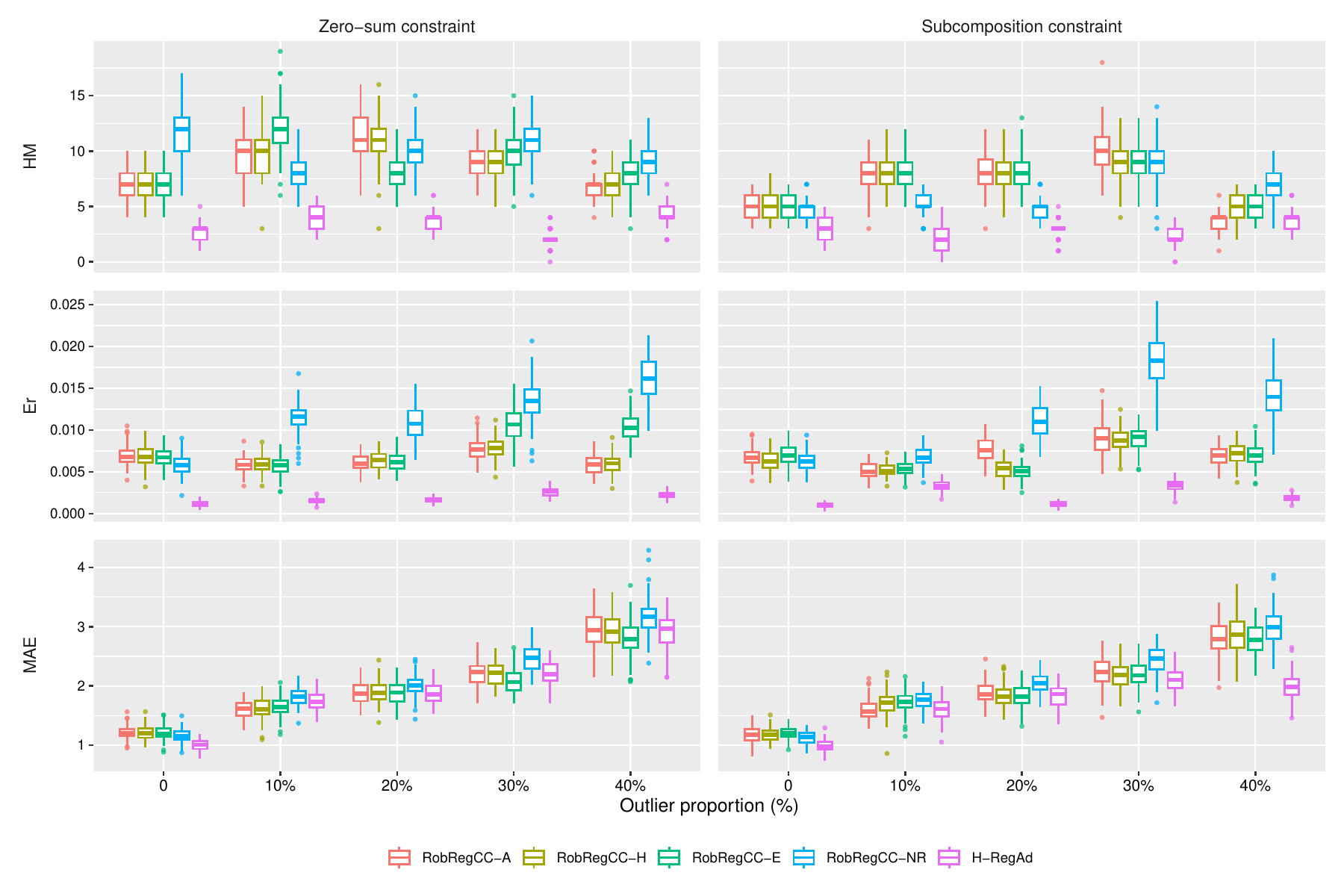}
		\caption{Simulation performance of H-RegAd and the competing methods
			in terms of HM, Er, and MAE ($L=0$, $n=200$, and $p=300$).}
		\label{fig:boxplots}
	\end{figure}
	
	
	To gain further insight into the effect of contamination on each method, we plot the trajectories of HM, Er, and MAE as functions of the outlier proportion under both the zero-sum and subcomposition constraints, as shown in  Figure~\ref{fig:traj_main}.
	As the contamination proportion increases from $0\%$ to $40\%$, we observe from the third row that HM and MAE tend to increase for several methods, which indicates that heavier contamination can make outlier identification and prediction more difficult. 
	In contrast, from the first and the second rows, we see that H-RegAd maintains relatively low HM and Er under both constraint structures, which suggests that it exhibits more stable outlier detection and coefficient estimation. 
	In a word, we conclude that these trajectory plots are consistent with Tables~\ref{tab:C1} and~\ref{tab:C2} and once again support the robustness of H-RegAd.
	
	\begin{figure}[!htbp]
		\centering
		\includegraphics[width=0.75\linewidth,height=0.82\textheight,keepaspectratio]{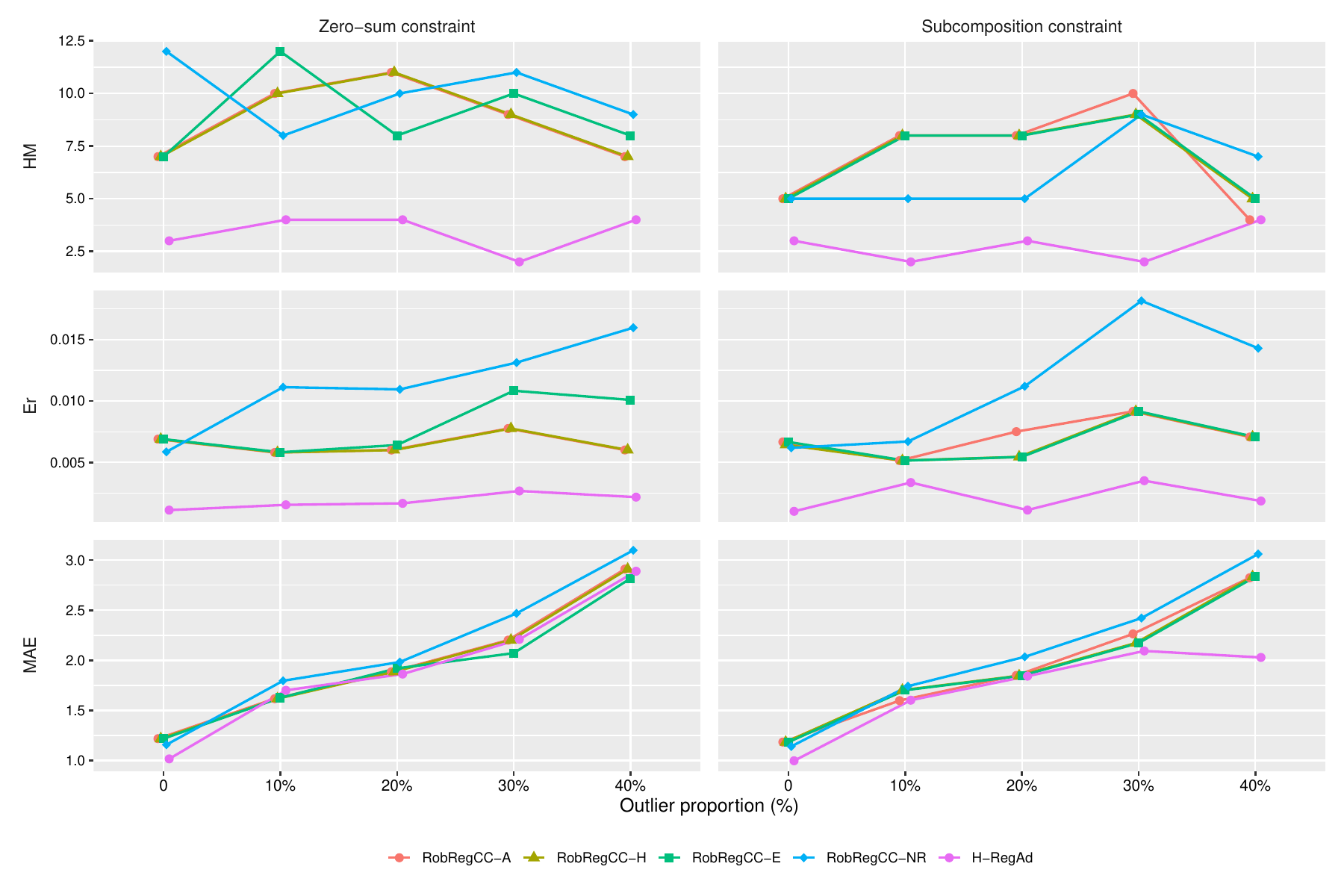}
		\caption{Trajectories of HM, Er and MAE under $L=0$ with zero-sum and subcomposition constraints.}
		\label{fig:traj_main}
	\end{figure}

	To further complement the performance comparisons in Figures \ref{fig:boxplots} and \ref{fig:traj_main}, we examine the convergence behavior of H-RegAd by plotting the dual KKT residuals under representative settings in Figure~\ref{fig:dualKKT}. 
	In this figure, the two columns correspond to the zero-sum and subcomposition constraints, while the two rows distinguish cases without and with leverage points. 
	In each plot, the dashed line indicates the stopping tolerance $10^{-4}$.
	Across all settings, we observe that the residual decreases rapidly in the early iterations and then stabilizes toward the prescribed tolerance. 
	We conclude that the employed sGS-ADMM algorithm used to implement H-RegAd  converges stably in practice under both constraint settings, even in the presence of influential contamination.
	
	\begin{figure}[!htbp]
		\centering
		\includegraphics[width=0.75\linewidth]{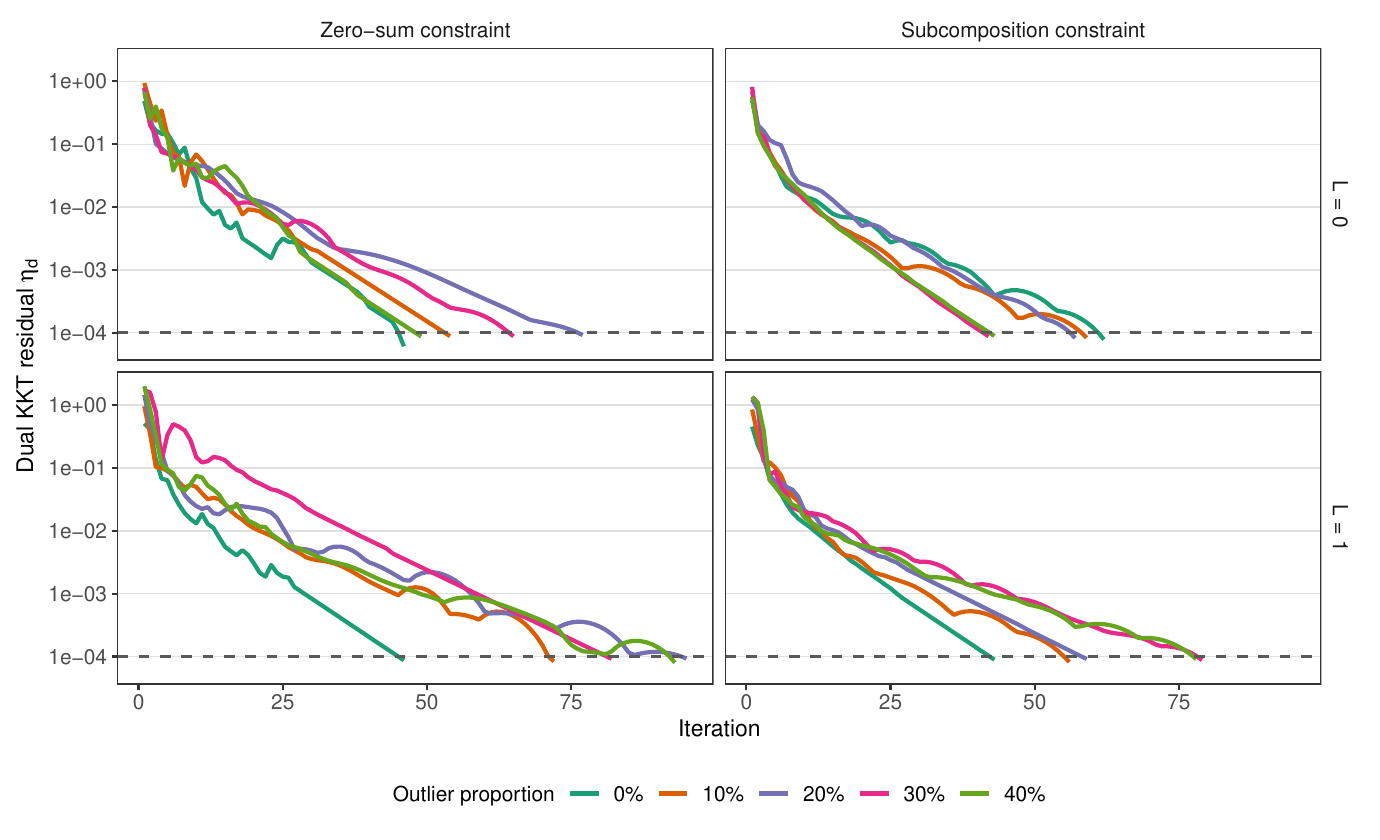}
		\caption{Convergence profiles of the dual KKT residual $\eta_d$ for
			the sGS-ADMM algorithm.}
		\label{fig:dualKKT}
	\end{figure}
	
	\section{Real Data Analysis: Gut Microbiome Data}
	\label{sec:realdata}
	In this section, we evaluate the practical performance of H-RegAd on a gut microbiome dataset of HIV patients. In this dataset, the response variable $y$ is soluble CD14 (sCD14), an immune marker linked to chronic inflammation and monocyte activation, 
	and the predictors are genus-level microbial relative abundances, with coefficient vector $\beta$.
	The dataset consists of $n = 151$ observations and $p = 60$ bacterial genera features.
	We note that this dataset has been adopted in \cite{rivera2018balances} as a benchmark example for the balance selection method selbal, which is a greedy step-by-step algorithm for sparse log-contrast modeling. For this reason, this real dataset serves as a good test case to verify whether our proposed robust framework can enhance model fitting and variable selection performance when atypical data observations exist.
	
	\subsection{Data Preprocessing and Model Formulation}
	The microbiome covariates in this dataset are compositional and contain zero values. Therefore, a standard preprocessing procedure was applied before model fitting. In this work, we replaced zero entries by $0.5$, after which each sample was normalized so that the components summed to one.
	Let $X \in \mathbb{R}^{n \times 60}$ denote the compositional data matrix. We construct the regression design matrix as
	$\widetilde{X} = \log(X)$.
	To account for the simplex structure of the predictors, we impose a zero-sum log-contrast constraint
	$\mathbf{1}_p^{\mathsf T}\beta = 0$.
	For the proposed method, the tuning parameter $\lambda$ was chosen by 10-fold cross-validation based on a scaled Huber loss.
	
	\subsection{Model Fitting and Predictive Performance}
	In this section, we evaluate model performance from three complementary perspectives: full-data goodness of fit, out-of-sample prediction accuracy, and model complexity.
	For goodness of fit, we use the coefficient of determination
	\[
	R^2 := 1 - \frac{\sum_{i=1}^n (y_i - \widehat{y}_i)^2}{\sum_{i=1}^n (y_i - \bar{y})^2},
	\]
	where $\bar{y} = \left ( 1/n \right ) \sum_{i=1}^n y_i$ denotes the sample mean of the response. The $R^2$ statistic measures the proportion of total variability in the response that is explained by the fitted model.
	Prediction accuracy is evaluated using the mean absolute error (MAE) and the median absolute error (MedAE), defined respectively as
\[
\mathrm{MAE}:=\frac{1}{n}\sum_{i=1}^n|y_i-\widehat{y}_i|,\qquad
\mathrm{MedAE}:=\mathrm{median}\Bigl\{|y_i-\widehat{y}_i|:i=1,\dots,n\Bigr\}.
\]
    Here, MAE measures the average magnitude of the residual errors, while
    MedAE is more robust to extreme residuals and thus better reflects the
    typical fitting error.
    Furthermore, to assess sparsity and outlier adjustment across each
    method, we report the number of selected taxa and potential outliers,
    which are used to reflect model complexity and the degree of outlier
    correction, respectively.
    
    	The above measures are first calculated on the full dataset to
    	characterize the goodness-of-fit of the fitted models. To further
    	evaluate out-of-sample predictive performance, we conduct a repeated
    	hold-out analysis following \citep{Mishra2022RobustCompositional}.
    	Specifically, the 151 observations are randomly split 100 times into
    	90\% training samples and 10\% test samples. All competing methods are
    	evaluated using the same random partitions, while model fitting and
    	parameter selection are performed using only the corresponding
    	training samples.
    In Table~\ref{tab:hiv_summary}, we summarize the full-data fitting
    performance, sparsity, and potential outlier detection results for all
    competing methods on the HIV microbiome dataset.
	
	\begin{table}[t]
		\centering
		\caption{Comparison of model fitting performance, sparsity, and
			potential outlier detection on the HIV microbiome dataset.}
		\label{tab:hiv_summary}
		
		\small
		\setlength{\tabcolsep}{3pt}
		\renewcommand{\arraystretch}{1.10}
		
		\begin{tabularx}{\textwidth}{
				@{}
				>{\raggedright\arraybackslash}p{0.21\textwidth}
				*{5}{>{\centering\arraybackslash}X}
				@{}
			}
			\toprule
			Method 
			& \shortstack{Selected\\Taxa}
			& \shortstack{Potential\\Outliers}
			& $R^2$ 
			& MAE 
			& MedAE \\
			\midrule
			
			RobRegCC-NR 
			& 10 
			& 0  
			& 0.274 
			& 1747.588 
			& 1431.693 \\
			
			RobRegCC-A  
			& 9  
			& 14 
			& 0.626 
			& 1416.271 
			& 1314.254 \\
			
			RobRegCC-E  
			& 9  
			& 16 
			& 0.533 
			& 1547.607 
			& 1383.283 \\
			
			RobRegCC-H  
			& 8  
			& 7  
			& 0.567 
			& 1476.908 
			& \textbf{1178.599} \\
			
			H-RegAd     
			& 16 
			& 30 
			& \textbf{0.687} 
			& \textbf{1295.068} 
			& 1219.779 \\
			
			\bottomrule
		\end{tabularx}
		
	\end{table}
	
	From Table~\ref{tab:hiv_summary}, we observe that H-RegAd provides the
	best full-data goodness-of-fit, yielding the largest $R^2$ and the
	smallest MAE among all methods.
	This suggests that H-RegAd provides a closer fit to the observed sCD14
	levels on the full dataset.
	In terms of median fitting error, RobRegCC-H attains the smallest MedAE,
	indicating good performance for typical observations.
	Regarding model complexity, H-RegAd selects more taxa and flags more
	potential outliers than the competing methods, suggesting that the
	improved full-data fit is accompanied by reduced sparsity and a more
	aggressive outlier-adjustment mechanism.
	
    To directly compare with the non-robust baseline, we plot the observed
    and fitted sCD14 values for H-RegAd and RobRegCC-NR in
    Figure~\ref{fig:hiv_fit}.
    In Figure~\ref{fig:hiv_fit}, the red solid line represents the fitted
    trend from H-RegAd, while the blue dashed line corresponds to NR.
    The red triangles indicate observations flagged as potential outliers
    by H-RegAd.
    Overall, the fitted values from H-RegAd follow the observed sCD14 levels
    more closely, particularly in the upper range.
    This is consistent with its higher full-data $R^2$ and lower full-data
    MAE reported in Table~\ref{tab:hiv_summary}.
    These results indicate that H-RegAd provides a closer fit to the
    observed data while accommodating observations that deviate from the
    dominant regression pattern.
	
	\begin{figure}[!htbp]
		\centering
		\includegraphics[width=0.65\linewidth]{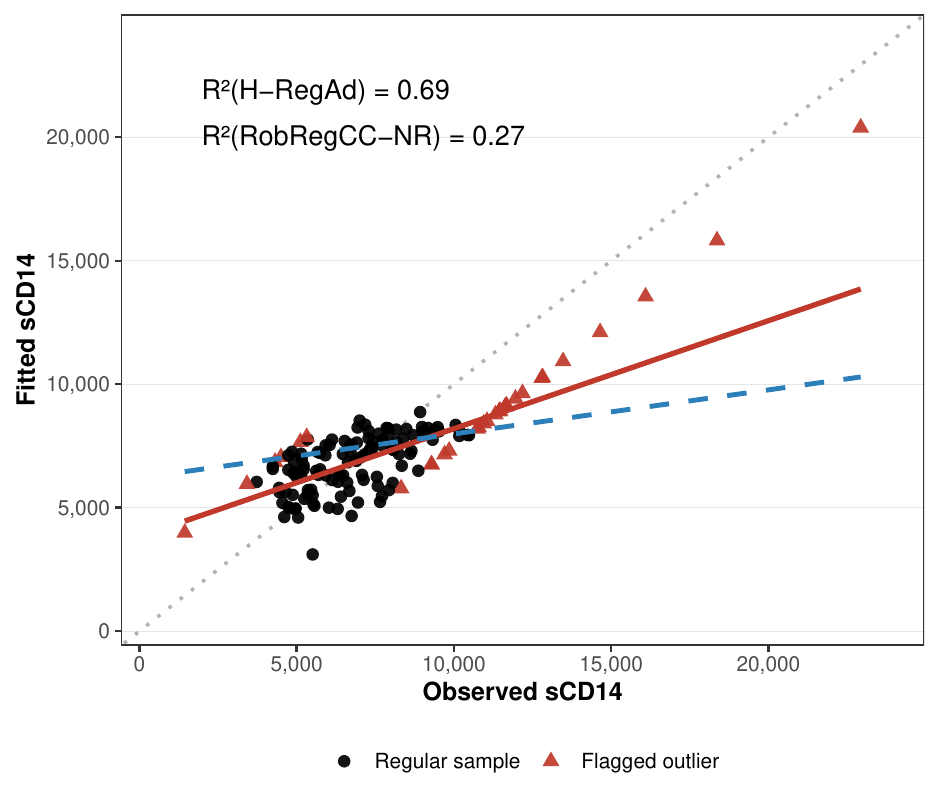}
		\caption{Observed versus fitted sCD14 values for H-RegAd and RobRegCC-NR on the HIV gut microbiome dataset.}
		\label{fig:hiv_fit}
	\end{figure}
	
	\subsection{Model-Selected Genera and Biological Interpretation}
	In this section, we summarize the microbial genera screened by the proposed H-RegAd method and explain the biological meaning of their corresponding regression coefficients.
	Under H-RegAd, we select a total of $16$ genera, among which $5$ yielded positive coefficients while the remaining $11$ yielded negative coefficients.
	Furthermore, in Table \ref{tab:hiv_genera}, we summarize representative genera selected by the proposed H-RegAd method, along with their regression coefficients and cross-model screening results.
	In addition, since we fit the model under the zero-sum log-contrast constraint, we therefore interpret these coefficients on the relative abundance scale rather than regarding them as simple marginal effects of absolute abundances.

	\begin{table}[!t]
		\centering
		\caption{Representative genera selected by H-RegAd}
		\label{tab:hiv_genera}
		\small
		\renewcommand{\arraystretch}{1.10}
		\setlength{\tabcolsep}{5pt}
		
		\begin{tabularx}{\textwidth}{
				@{}
				>{\raggedright\arraybackslash}X
				>{\centering\arraybackslash}p{0.16\textwidth}
				>{\raggedright\arraybackslash}p{0.36\textwidth}
				@{}
			}
			\toprule
			Taxon & Coefficient & Selected by \\
			\midrule
			
			g\_\textit{Collinsella}
			& $-629.470$
			& RobRegCC-NR, RobRegCC-A, RobRegCC-E, RobRegCC-H, H-RegAd \\
			
			g\_\textit{Faecalibacterium}
			& $353.445$
			& RobRegCC-NR, RobRegCC-A, RobRegCC-E, H-RegAd \\
			
			f\_\textit{Lachnospiraceae}\_g\_\textit{Incertae Sedis}
			& $260.933$
			& RobRegCC-NR, RobRegCC-A, RobRegCC-E, RobRegCC-H, H-RegAd \\
			
			f\_\textit{Ruminococcaceae}\_g\_unclassified
			& $258.224$
			& RobRegCC-NR, H-RegAd \\
			
			g\_\textit{Bacteroides}
			& $218.977$
			& H-RegAd \\
			
			o\_\textit{Clostridiales}\_g\_unclassified
			& $-111.576$
			& H-RegAd \\
			
			f\_\textit{Porphyromonadaceae}\_g\_\allowbreak unclassified
			& $-110.157$
			& H-RegAd \\
			
			g\_\textit{Prevotella}
			& $82.402$
			& H-RegAd \\
			
			\bottomrule
		\end{tabularx}
	\end{table}
	
   We observe from Table~\ref{tab:hiv_genera} that the selected microbial
   taxa exhibit distinct association patterns with sCD14.
   Specifically, g\_\textit{Collinsella} yields the strongest negative
   coefficient ($-629.470$), while g\_\textit{Faecalibacterium} has the
   largest positive coefficient ($353.445$).
   We also observe that H-RegAd selects several microbial taxa that are
   consistently identified by the competing methods, including
   g\_\textit{Collinsella} and
   f\_\textit{Lachnospiraceae}\_g\_\textit{Incertae\_Sedis}, suggesting
   that these associations are relatively stable across different model
   specifications.
   In addition, H-RegAd uniquely selects four candidate taxa,
   \textit{g\_Bacteroides},
   o\_\textit{Clostridiales}\_g\_\textit{unclassified},
   f\_\textit{Porphyromonadaceae}\_g\_\textit{unclassified}, and
   g\_\textit{Prevotella}, which are not selected by the other competing
   methods.
   These results suggest that H-RegAd may capture additional candidate
   microbial associations with sCD14 while retaining several signals that
   are also identified by the comparison methods.
	
	To fully display the regression coefficients of the genera selected by
	H-RegAd, we construct the horizontal bar plot shown in
	Figure~\ref{fig:hiv_selected_genera}, with red bars denoting positive
	coefficients and blue bars denoting negative coefficients.
	We observe that taxa with negative coefficients account for the majority
	of the selected features, indicating that H-RegAd identifies a relatively
	large number of taxa with negative associations with sCD14 in the
	relative abundance space.
	Consistent with the results reported above, H-RegAd selects a relatively
	diverse set of microbial features while retaining several taxa that are
	also identified by the competing methods.
	
	\begin{figure}[!t]
		\centering
		\includegraphics[width=0.65\linewidth]{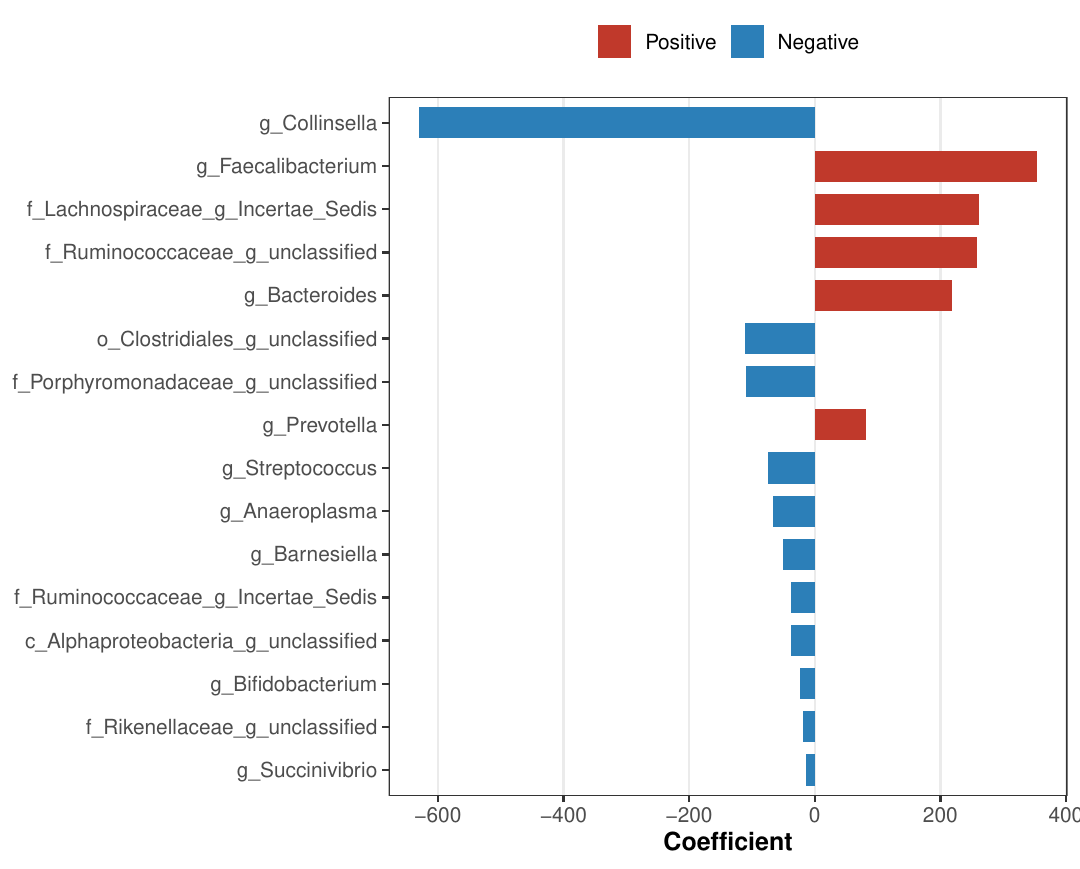}
		\caption{Estimated coefficients of genera selected by H-RegAd on the HIV microbiome data.}
		\label{fig:hiv_selected_genera}
	\end{figure}
	
	\subsection{Sensitivity Analysis to Pseudocount Selection}
	To assess whether the empirical findings are sensitive to the choice
	of zero replacement, we repeated the H-RegAd analysis using
	pseudocounts of $0.1$, $0.5$, and $1.0$ prior to total-sum
	normalization. Except for the pseudocount value, all other
	preprocessing steps, model settings, and tuning procedures were kept
	unchanged. The pseudocount of $0.5$ corresponds to the preprocessing
	scheme used in the main analysis.

\begin{table}[!t]
	\centering
	\caption{Sensitivity analysis of H-RegAd under different pseudocounts.}
	\label{tab:pseudocount_sensitivity}
	
	\footnotesize
	\setlength{\tabcolsep}{3pt}
	\renewcommand{\arraystretch}{1.15}
	
	\begin{tabularx}{\textwidth}{
			@{}
			*{6}{>{\centering\arraybackslash}X}
			@{}
		}
		\toprule
		Pseudocount
		& \shortstack{Selected\\taxa}
		& $R^2$
		& MAE
		& \shortstack{Overlap\\with $0.5$}
		& \shortstack{Sign\\consistency} \\
		\midrule
		
		0.1 & 18 & 0.706 & 1237.869 & 14/16 & 100\% \\
		0.5 & 16 & 0.687 & 1295.068 & 16/16 & 100\% \\
		1.0 & 13 & 0.613 & 1419.053 & 12/16 & 100\% \\
		
		\bottomrule
	\end{tabularx}
\end{table}
	
	As shown in Table~\ref{tab:pseudocount_sensitivity}, H-RegAd selected
	18, 16, and 13 taxa under pseudocounts of $0.1$, $0.5$, and $1.0$,
	respectively. Relative to the baseline pseudocount of $0.5$, 14 of the
	16 selected taxa were retained when using $0.1$, whereas 12 were
	retained when using $1.0$. Importantly, the coefficient signs of all
	commonly selected taxa were fully consistent across the three
	pseudocount settings. The coefficient correlations were also relatively
	high between $0.1$ and $0.5$ ($0.937$) and between $0.5$ and $1.0$
	($0.886$). Although the degree of sparsity and the full-data fitting
	measures varied to some extent with the pseudocount, the core selected
	taxa and the directions of their estimated associations remained
	largely consistent.
	
	Finally, we conduct repeated subsampling tests to further verify the stability of the selected genera under sample perturbation. The corresponding results are presented in the Appendix: the core genera identified using the full dataset can be consistently detected across most subsamples. Full details of this supplementary analysis are provided in  \ref{Appensec}.

	\section{Conclusions}
	\label{sec:discussion}
	In this paper, we propose H-RegAd for robust sparse regression with compositional covariates under linear structural constraints.
	The proposed method combines the Huber loss function with a Lasso penalty, which allows it to perform variable selection while resisting atypical observations.
	The framework accommodated both zero-sum log-contrast constraints and more general subcomposition constraints.
	We derived the dual problem corresponding to the constrained optimization objective of H-RegAd and adopted the sGS-ADMM algorithm to efficiently solve this dual formulation.
	Furthermore, we conducted extensive numerical experiments on both simulated and real datasets to demonstrate the efficacy of the proposed method.
	In simulation studies, we found that H-RegAd achieved high estimation accuracy, competitive prediction performance, and reliable outlier identification under various contamination proportions and constraint settings.
	In the HIV microbiome analysis,  H-RegAd achieved the best overall fit and identified an interpretable set of microbial biomarkers, and complementary stability analyses further confirm the reproducibility of the main selected genera.
	In a word, we believe that the proposed H-RegAd offered an effective and interpretable framework for robust regression analysis with compositional predictors subject to structural linear constraints.
	
	\section*{Statements and Declarations}
	
	\subsection*{Acknowledgements}
	The work of Yunhai Xiao is supported by the National Natural Science
	Foundation of China (Grant Nos. 12471307 and 12271217).
	
	\subsection*{Conflict of Interest Statement}
	The authors declare that they have no relevant financial or
	non-financial interests to disclose.
	
	\subsection*{Data Availability}
	The gut microbiome dataset analyzed in this study is publicly available
	from the GitHub repository accompanying Mishra and Müller (2022) at
	\url{https://github.com/amishra-stats/robregcc}.
	The dataset was previously analyzed by Rivera-Pinto et al. (2018).
	No new empirical data were collected for this study.
	
	\subsection*{Code Availability}
	The R implementation of H-RegAd, together with the code used for the
	numerical experiments, is publicly available at
	\url{https://github.com/XukeHua/H-RegAd}.
	
	\subsection*{Author Contributions}
	All authors contributed to the study conception and design.
	All authors read and approved the final work.
	
\bibliographystyle{unsrtnat}
\bibliography{csda-references}

	\newpage
	\clearpage
	\appendix
\section{Additional Results}\label{Appensec}
	\addcontentsline{toc}{section}{Appendix}
	\subsection{Additional Performance Results and Convergence Diagnostics}
	Figure \ref{fig:boxplotsL1} presents the boxplots of HM, Er, and MAE under	$L=1$, where leverage points are introduced. 
	From Figure \ref{fig:boxplotsL1}, we observe that HM, Er, and MAE generally	increase as the contamination proportion becomes larger across several	methods. 
	Compared with the case of $L=0$, we see that the variability of the competing	methods becomes more pronounced, particularly under the subcomposition
	constraint, which indicates that leverage points pose additional challenges	for outlier identification, coefficient recovery, and prediction. 
	In	contrast, H-RegAd consistently maintains relatively low HM and Er under	both constraint structures. 
	Overall, we conclude that although MAE tends to increase as	the contamination becomes more severe, H-RegAd remains competitive in prediction accuracy in most settings, which further confirms its robustness in the presence of leverage points.
	\begin{figure}[!htbp]
		\centering
		\includegraphics[width=0.65\linewidth]{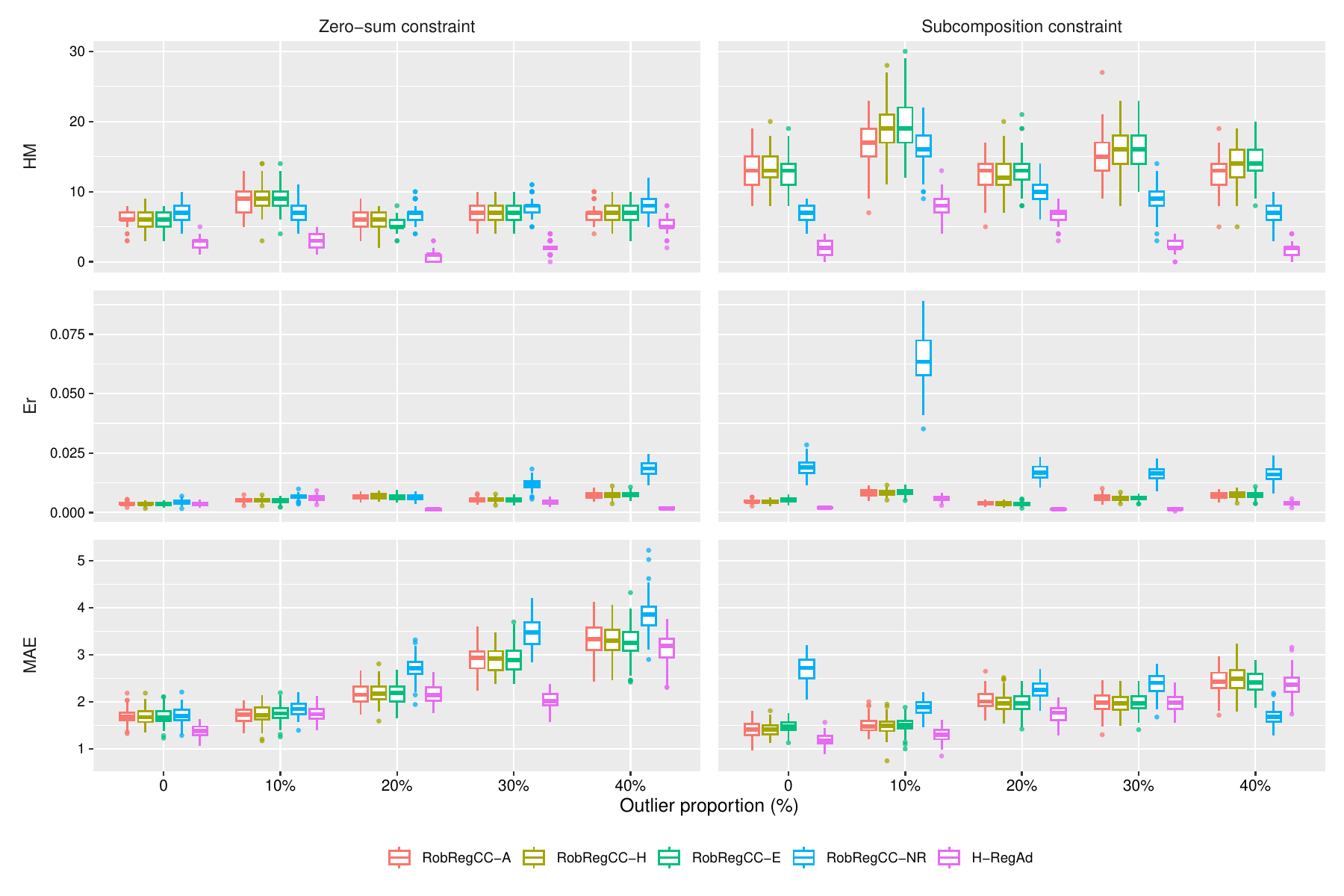}
		\caption{Simulation performance of H-RegAd and the competing methods
			in terms of HM, Er, and MAE ($L=1$, $n=200$, and $p=300$).}
		\label{fig:boxplotsL1}
	\end{figure}

	To gain further insight into the effect of leverage points on each method,	Figure~\ref{fig:traj_main1} presents the trajectories of HM, Er, and MAE
	as functions of the outlier proportion under $L=1$. 
	We see that, as the contamination proportion increases from $0\%$ to $40\%$, HM and MAE tend to increase for several competing methods, with more pronounced fluctuations under	the subcomposition constraint. 
	Compared with the case $L=0$, these	results indicate that leverage points pose additional challenges for outlier identification, coefficient estimation, and prediction. 
	In	contrast, H-RegAd maintains relatively low HM and Er under both constraint	structures and remains competitive in terms of MAE. 
	Overall, we find that the trajectory	plots are consistent with the boxplot results in	Figure \ref{fig:boxplotsL1} and provide further evidence of the robustness of H-RegAd in the presence of leverage points.
	\begin{figure}[!htbp]
		\centering
		\includegraphics[width=0.65\linewidth,height=0.82\textheight,keepaspectratio]{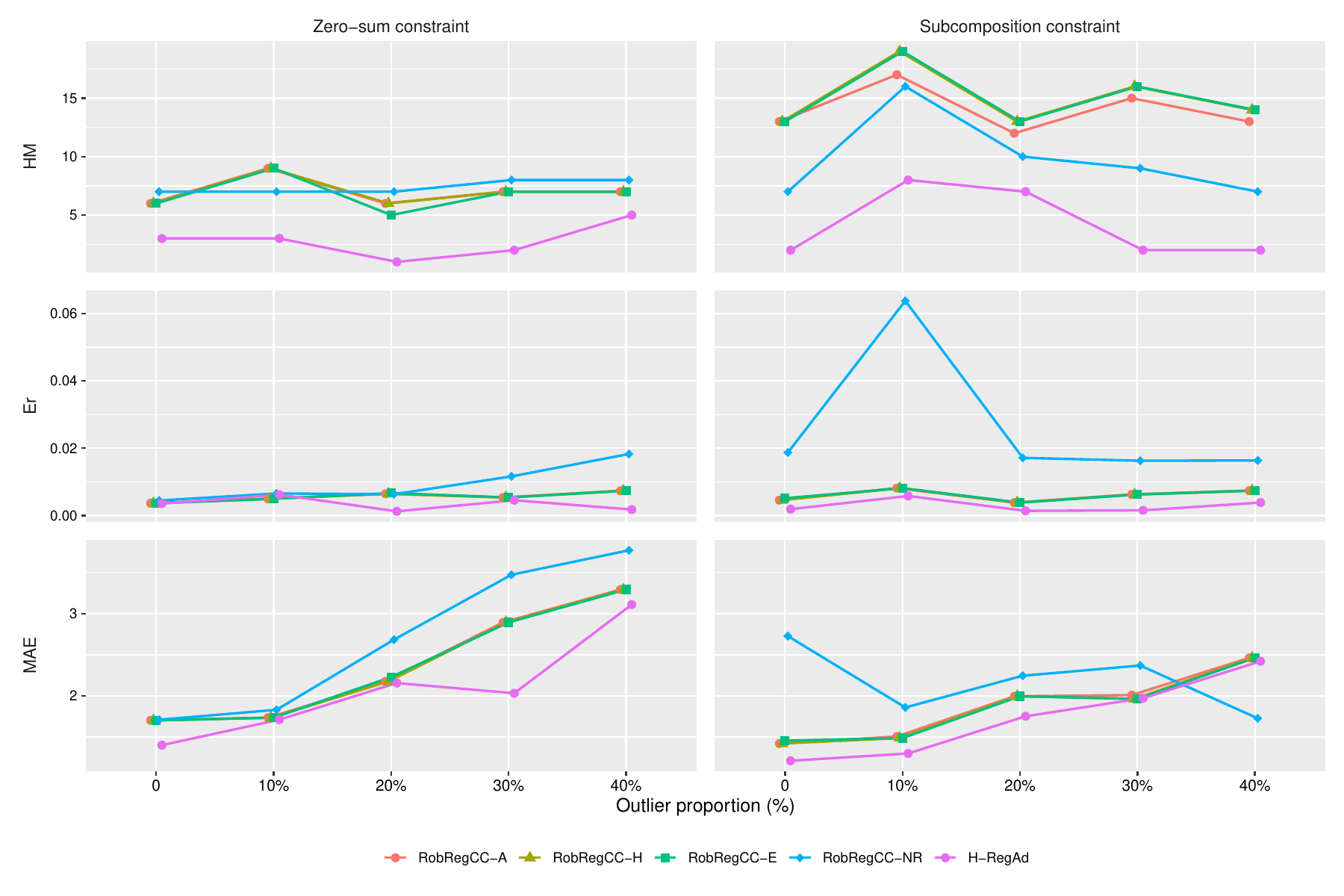}
		\caption{Trajectories of HM, Er, and MAE under \(L=1\) with zero-sum and subcomposition constraints.}
		\label{fig:traj_main1}
	\end{figure}
	
	\subsection{Stability Analysis of The Selected Genera}
	Because variable selection based on a single model fit may be sensitive to sampling variability, we further assess the stability of the selected genera through repeated subsampling. 
	Specifically, in this test, the random subsamples were repeatedly drawn from the HIV microbiome dataset, and H-RegAd was refitted to each subsample using the same modeling framework. 
	The selection frequency of each genus was then calculated as the proportion of subsamples in which it was selected. 
	This analysis helps distinguish genera that are consistently associated with the response from those selected only occasionally in the full-sample analysis.
	
	We summarize the representative genera with high selection frequencies under H-RegAd in Table \ref{tab:hiv_stability_app}. 
	In particular, we see that \textit{Collinsella}, \textit{Faecalibacterium}, \textit{Bacteroides}, and an unclassified \textit{Ruminococcaceae} taxon were selected in all subsamples, indicating that their inclusion is highly stable with respect to sampling variability. 
	Meanwhile, the remaining genera were also selected relatively frequently, with selection frequencies ranging from $0.802$ to $0.981$. 
	Moreover, we see that the coefficient signs of the frequently selected genera remained consistent across subsamples, suggesting that both variable selection and the direction of the estimated associations are stable under repeated sampling.
	\begin{table}[!t]
		\centering
		\caption{Representative genera with high selection frequencies under H-RegAd}
		\label{tab:hiv_stability_app}
		
		\footnotesize
		\setlength{\tabcolsep}{4pt}
		\renewcommand{\arraystretch}{1.05}
		
		\begin{tabularx}{\textwidth}{
				@{}
				>{\raggedright\arraybackslash}X
				>{\centering\arraybackslash}p{0.17\textwidth}
				>{\centering\arraybackslash}p{0.18\textwidth}
				@{}
			}
			\toprule
			Taxon
			& \shortstack{Selection\\frequency}
			& \shortstack{Mean\\coefficient} \\
			\midrule
			
			g\_\textit{Collinsella}
			& 1.000 & $-546.990$ \\
			
			g\_\textit{Faecalibacterium}
			& 1.000 & $383.401$ \\
			
			g\_\textit{Bacteroides}
			& 1.000 & $328.225$ \\
			
			unclassified \textit{Ruminococcaceae}
			& 1.000 & $244.661$ \\
			
			unclassified \textit{Porphyromonadaceae}
			& 0.981 & $-149.527$ \\
			
			\textit{Barnesiella}
			& 0.963 & $-53.641$ \\
			
			\textit{Prevotella}
			& 0.942 & $84.781$ \\
			
			unclassified \textit{Clostridiales}
			& 0.864 & $-95.951$ \\
			
			unclassified \textit{Alphaproteobacteria}
			& 0.844 & $-111.991$ \\
			
			\textit{Lachnospiraceae} \textit{Incertae Sedis}
			& 0.802 & $185.420$ \\
			
			\bottomrule
		\end{tabularx}
	\end{table}
	We display the selection frequencies of representative genera under repeated subsampling in Figure~\ref{fig:hiv_stability_app}, with the bar colors indicating the signs of their mean coefficients. 
	We see from this figure that several genera, including \textit{Collinsella}, \textit{Faecalibacterium}, \textit{Bacteroides}, and an unclassified \textit{Ruminococcaceae} taxon, are selected in all subsamples. 
	Meanwhile, most of the remaining genera also exhibit relatively high selection frequencies, whereas \textit{Aloprevotella} and \textit{Escherichia--Shigella} are selected less frequently. 
	In addition, both positive and negative coefficient signs are observed among the frequently selected genera, indicating that H-RegAd identifies stable microbial associations in both directions. 
	Overall, these results complement those in Table \ref{tab:hiv_stability_app} and provide further evidence that the selected microbial signals are robust to sampling variability.

	\begin{figure}[!htbp]
		\centering
		\includegraphics[width=0.65\linewidth]{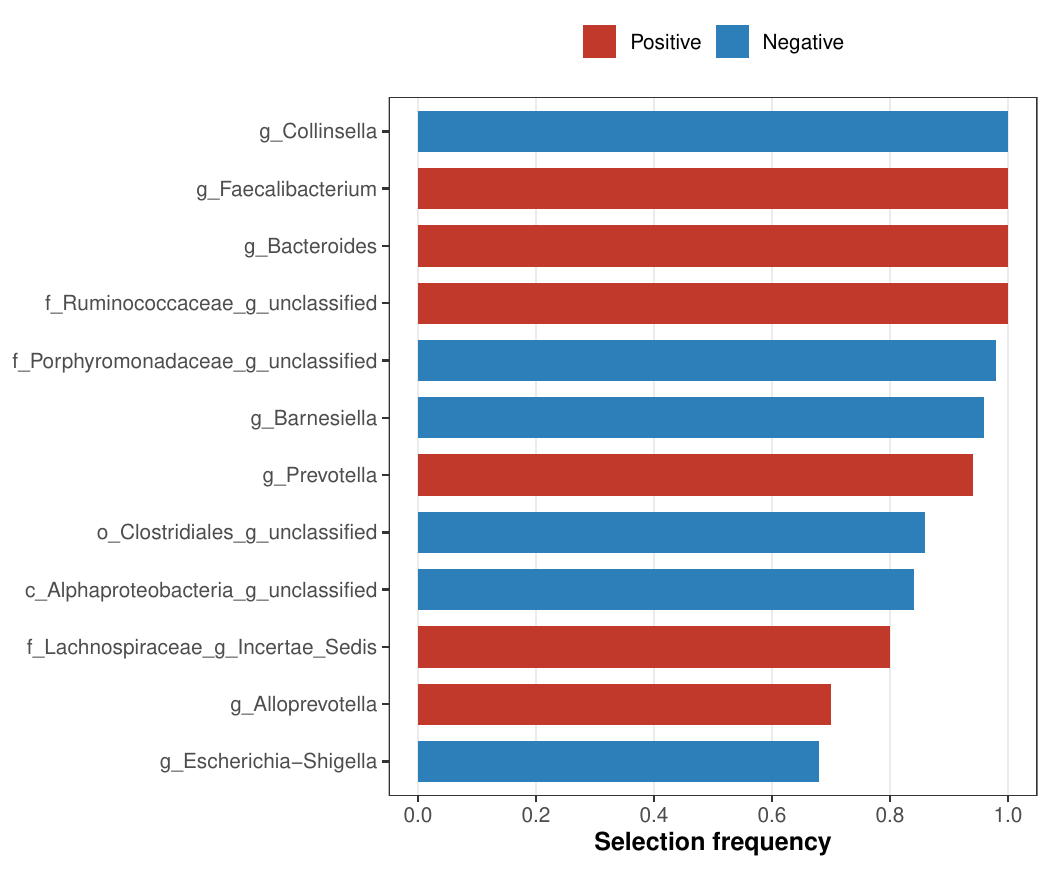}
		\caption{Selection frequencies of representative genera under repeated subsampling for H-RegAd.}
		\label{fig:hiv_stability_app}
	\end{figure}

\end{document}